\documentclass[11pt,a4paper]{article}
 \usepackage[english]{babel} 
 \usepackage{amsmath}
 \usepackage{mathtools}
 \usepackage{amsfonts}
 \usepackage{amssymb}
 \usepackage{amsthm}
 \usepackage{makeidx}
 \usepackage{graphicx}
 \usepackage{natbib}
 \usepackage{setspace}
 \usepackage{float}
 \usepackage{hyperref}
 \usepackage[margin=1in]{geometry}
 \usepackage{setspace}
 \usepackage[titletoc]{appendix}

 \newcommand{\pcon}{\ensuremath{ {\rho^\mu_{conv}} }}
 \newcommand{\A}{\ensuremath{\mathbb{A} }}
 \newcommand{\Iu}{\ensuremath{ {\mathcal{I}_\mu }}}
 \newcommand{\I}{\ensuremath{ {\mathcal{I}} }}
 \newcommand{\p}{\ensuremath{ {\rho^i} }}
 
 \hypersetup{
 	colorlinks   = true,
 	citecolor    = blue,
 	urlcolor = blue
 }
 
\author{Marcelo Brutti Righi$^{a,}$\footnote{Corresponding author. We thank the Lead Guest Editor, Professor Claudio Fontana, and the two anonymous Reviewers for their constructive comments, which have been
		useful to improve both the technical quality and presentation of the manuscript. We are grateful for the financial support of FAPERGS (Rio Grande do Sul State
		Research Council) project number 17/2551-0000862-6 and CNPq (Brazilian Research Council) projects number
		302369/2018-0 and 407556/2018-4.}\\  \small{ \href{mailto:marcelo.righi@ufrgs.br}{marcelo.righi@ufrgs.br}} 
	\and Marlon Ruoso Moresco$^{a}$\\\small{\href{mailto:marlon.moresco@ufrgs.br}{marlon.moresco@ufrgs.br}}}

\date{\small{$^{a}$\textit{Business School, Federal University of Rio Grande do Sul, Washington Luiz, 855, Porto Alegre, Brazil, zip 90010-460}}}

 \title{Inf-convolution and optimal risk sharing with countable sets of risk measures}

 \newtheorem{Def}{Definition}[section]
 \newtheorem{Thm}[Def]{Theorem}
 \newtheorem{Prp}[Def]{Proposition}
 \newtheorem{Lmm}[Def]{Lemma}
 \newtheorem{Crl}[Def]{Corollary}
 \theoremstyle{definition}
 \newtheorem{Exm}[Def]{Example}
 
 \theoremstyle{remark}
 \newtheorem{Rmk}[Def]{Remark}

 \numberwithin{equation}{section}

\begin{document}

 	\maketitle
 	\begin{abstract}
  		The inf-convolution of risk measures is directly related to risk sharing and general equilibrium, and it has attracted considerable attention in mathematical finance and insurance problems. However, the theory is restricted to finite sets of risk measures. This study extends the inf-convolution of risk measures in its convex-combination form to a countable (not necessarily finite) set of alternatives. The intuitive meaning of this approach is to represent a generalization of the current finite convex weights to the countable case. Subsequently, we extensively generalize known properties and results to this framework. Specifically, we investigate the preservation of properties, dual representations, optimal allocations, and self-convolution.
 	\end{abstract}
 	\section{Introduction}
 	
 The theory of risk measures has attracted considerable attention in mathematical finance and insurance since the seminal paper by \cite{Artzner1999}. The books by \cite{Pflug2007}, \cite{Delbaen2012}, \cite{Ruschendor2013}, and \cite{Follmer2016} are comprehensive expositions of this subject. In these studies, a key topic is the inf-convolution of risk measures directly related to risk sharing and general equilibrium. These problems may be connected with regulatory
 capital reduction, risk transfer in insurance--reinsurance
 contracts, and several other applications in  classic studies such as \cite{Borch1962}, \cite{Arrow1963}, \cite{Gerber1978}, and \cite{Buhlmann1982}, as well as more recent research as in \cite{Landsberger1994}, \cite{Dana2003}, and \cite{Heath2004}. 
 
 Formally, the inf-convolution of risk measures is defined as \[\square_{i=1}^n\rho^i(X)=\inf\left\lbrace \sum_{i=1}^{n}\rho^i(X^i)\colon\sum_{i=1}^{n}X^i=X\right\rbrace,\] where $X$ and $X^i,\:i=1,\cdots,n$, belong to some linear space of random variables over a probability space, and $\rho^i,\:i=1,\cdots,n$, are risk measures, which are functionals on this linear space.

 Convex risk measures, as initially proposed by \cite{Follmer2002} and \cite{Frittelli2002}, have recently attracted considerable attention in the context of inf-convolutions, as in several other areas of risk management. This subject is explored in \cite{Barrieu2005}, \cite{Burgert2006}, \cite{Burgert2008}, \cite{Jouini2008}, \cite{Filipovic2008}, \cite{Ludkovski2008}, \cite{Ludkovski2009},  \cite{Acciaio2009}, \cite{Acciaio2009b}, \cite{Tsanakas2009}, \cite{Dana2010}, \cite{Delbaen2012}, and \cite{Kazi-Tani2017}. These studies present a detailed investigation of the properties of inf-convolution as a risk measure per se, as well as optimality conditions for the resulting allocations. 
 
 Beyond the usual approach of convex risk measures, some studies have been concerned with inf-convolution in relation to specific properties, as in \cite{Acciaio2007}, \cite{Grechuk2009}, \cite{Grechuk2012}, \cite{Carlier2012}, \cite{Mastrogiacomo2015}, and \cite{Liu2019}, particular risk measures, as the recent quantile risk sharing in \cite{Embrechts2018}, \cite{Embrechts2018b}, \cite{Weber2018}, \cite{Wang2018}, and \cite{Liu2019b}, or even specific topics, as in \cite{Liebrich2019}. Other recent papers dealing, directly or indirectly, with (classical) inf-convolutions are \cite{Bellini2021}, \cite{Burzoni2022}, \cite{Castagnoli2021}, \cite{Kirilyuk2021}, \cite{Liebrich2021} and \cite{Wang2021}.
 
 However, these studies are restricted to finite sets of risk measures. As observed by \cite{Tsanakas2009}, for convex but not positively homogeneous risk measures and without market frictions like transaction costs, risk can usually be reduced arbitrarily by introducing more subsidiaries, and hence there is no incentive to stop this splitting procedure.  Under uncertainty and without limitations on which lines to open for risk sharing, using the inf-convolution with $n \to \infty$ is helpful to get the best allocation and then open the relevant business lines, as it considers all possible divisions.
 
 Thus, in this study, we seek to extend the inf-convolution of risk measures to an infinite countable set of alternatives. Specifically, we consider a collection of risk measures $\rho_\mathcal{I}=\{\rho^i,\:i\in\mathcal{I}\}$, where $\mathcal{I}$ is a nonempty countable set.  In this sense, we can think into considering a functional  as 	\[
 \rho_{conv}(X)=\inf\left\lbrace \sum_{i\in \mathcal{I}}\rho^i(X^i)\colon\sum_{i\in \mathcal{I}}X^i=X\right\rbrace. \]  It is the minimum amount of risk, which may represent capital requirement, for instance, among all possible ways
 of dividing a risk $X$ into countable fragments and distributing the capital into countable units. Such units could be business lines, agents, lotteries, liquidation times, etc. Thus, such a formulation arises naturally when the agent can pulverize its position in as many fragments as desired. The critical point to be noted is that the split of the position $X$ can be taken at any number of fragments as desired instead of a fixed finite one as in the usual approach. Thus, the value generated for the resulting risk measure is expected to be smaller than the one resulting from any fixed finite inf-convolution.
 
 However,  such adaptation is not possible from a technical point of view since it is ill definite 
 since we cannot assure the convergence of the infinite series $\sum_{i\in \mathcal{I}}\rho^i(X^i)$, which will typically diverge. This is the case in \cite{Tsanakas2009}, \cite{Wang2016} and \cite{Liebrich2019}, where one can easily ends with a pathological $-\infty$ value for risk. Hence, it is necessary to make adjustments in order to guarantee the well definiteness of the functional. More specifically, we have to make three adjustments: (i) to consider weights in the risk summation, (ii) to consider weights in the allocations, and (iii) to consider bounded allocations. In what follows, we explain the mathematical and financial/economic reasoning of such adjustments.
 
We thus use a slightly modified version, with convex combinations instead of simple sum, representing weighting schemes. In the finite case it may be considered as  $\mu=\{\mu_1,\cdots,\mu_n\}\in[0,1]^n,\:\sum_{i=1}^{n}\mu_i=1$; this modified version is defined as \[\rho^{\mu,n}_{conv}(X)=\inf\left\lbrace \sum_{i=1}^{n}\mu_i\rho^i(X^i)\colon\sum_{i=1}^{n}X^i=X\right\rbrace .\]
 Letting $\hat{\rho}^i=\mu_i\rho^i,\:i=1,\cdots,n$ immediately implies that $\rho^{\mu,n}_{conv}(X)$ is a special case of the standard $\square_{i=1}^n \hat{\rho}^i$. \cite{Starr2011} and \cite{Ravanelli2014} show that one needs to solve such formulation in general equilibrium theory in order to obtain all Pareto-optimal  allocations. In this context, $\mu$ are also called Negishi weights, which represent the importance of each $\rho^i$ in the global decision.
 
 Nonetheless, this map is not a monetary risk measure since it fails the usual cash additivity (Translation invariance) property. In \cite{Starr2011} this is not a problem since it deals with usual utilities that do not fulfill this property in general. However, this property is crucial for risk measures. In order to circumvent this issue without jeopardizing the convergence, we consider weights on both summations of risk measures and allocations. In financial matters, such a weighting scheme in the allocations is akin to portfolio weights.  In fact, in the traditional approach, we have that allocations are not divisible. However, this indivisibility assumption is not a natural one as we could
 consider portfolios. 
 Furthermore, notice that the ordering of weighting in the allocations $X^i$ is more flexible than the risk measures since we can freely vary the partition for indicator $i$, while the risk measures are fixed for each $i$. 
 
  We recognize that, despite the previously exposed reasoning, the weights in both risk summation and allocations are introduced in order to ensure nice mathematical properties while their financial and economical motivations are still elusive. In this sense, we left more concrete economical discussion and applications for future research. 
  
 Furthermore, in the countable case, we also need to consider bounded allocations to obtain the desired convergence to guarantee the well definiteness of the convolution.  Note that this is always the case for finite $\mathcal{I}$. Nonetheless, it also has a financial explanation. There may be either regulatory or legal impediments for pooling all risks, which creates bounds on $\{X_i\}$. Such bounds can represent the limit loss a position may support without collapsing the whole institution. Moreover, it may be difficult to create some portfolios, especially those containing risks traded in illiquid markets. Our side constraints are connected with a similar idea in portfolio theory, where one considers bounded admissible positions to exclude strategies that allow arbitrage. Moreover, the consideration of constrained allocations is not new in the literature. For instance, \cite{Burgert2006} and \cite{Burgert2008} study a constrained inf-convolution where $X\geq 0$ (or  $X\leq 0$) implies that  $X^i\geq 0 $ (or $ X^i\leq 0$)  for any $i$ in the allocation. \cite{Wang2020} consider bounded positions for their portfolios an extreme-aggregation measure. Of course, such constraints would hold for the convex weighted allocations.

 Under all such reasoning, in this study, we extend the convex combination-based inf-convolution of risk measures to an infinite countable set of alternatives as follows:
 \[
 \rho^\mu_{conv}(X)=\inf\left\lbrace \sum_{i\in \mathcal{I}}\rho^i(X^i)\mu_i\colon\sum_{i\in \mathcal{I}}X^i\mu_i=X,\:\{X^i\}\:\text{is bounded}\right\rbrace.\] 
 The intuitive principle of this approach is to regard $\mu=\{\mu_i\}_{i\in\mathcal{I}}\subset[0,1]$ such that $\sum_{i\in \mathcal{I}}\mu_i=1$ as a generalization of convex weights in the finite case.We extensively generalize known properties and results to this framework. More specifically, we investigate the preservation of properties of $\rho_{\mathcal{I}}$, dual representations, optimal allocations, and self-convolution. Of course, we do not intend to be exhaustive due to the extent of the related literature. To the best of our knowledge, there is no study in this direction. Furthermore, we can represent our functional in the same vein as the traditional summation under some suitable choices for the conversion. In particular, the two frameworks coincide and are well suited under Positive homogeneity.

 Regarding the related literature, the work of \cite{Wang2016} considers a countable allocation, but with $\rho^i=\rho$ for any $i\in\mathcal{I}$, i.e., a fixed risk measure.
 {In his paper,  regulatory arbitrage occurs when dividing a position into several fragments results in a reduced capital requirement to the joint position.}
 Similar discussions are provided in \cite{Tsanakas2009} and \cite{Liebrich2019}, where it is proved that transaction costs should not be neglected in such a setting since one can quickly end up with a pathological $-\infty$ value for risk. In this sense, these authors study the problem of finding cost-optimal portfolio splits under market frictions. In our framework, we generalize such reasoning by allowing different risk measures. Thus, our approach identifies the limit case regarding all possibilities for the agent concerning the division of a position. We also allow for weights in order to allow for distinct degrees of importance for each risk measure. Furthermore, our approach is well defined and directly related to the usual summation approach, even without transaction costs. Thus, for ease of discussion, we assume that no cost occurs in dividing risks.

 The study by \cite{Righi2019} considers an arbitrary set of risk measures and investigates the properties of combinations of the form $\rho=f(\rho_\mathcal{I})$, where $f$ is a combination function over a linear space generated by the outcomes of $\rho_{\mathcal{I}}(X)=\{\rho^i(X),\:i\in\mathcal{I}\}$. Under the lack of a universal choice of the best risk measure from a set of alternatives, one can consider the use of many candidates to benefit from distinct qualities. The inf-convolution is not suitable for such a framework {with} fixed $X$. Under our approach for a countable set of candidates, it is possible to split the position to obtain the best allocation.
 
 From a mathematical point of view,  the countable $\mathcal{I}$ is a limiting case. The consideration of an arbitrary (not necessarily finite or countable) set of risk measures could be  done by considering a measure (probability) $\mu$ over a suitable sigma algebra $\mathcal{G}$ of $\mathcal{I}$. The problem then becomes
 \[
 \rho^\mu_{conv}(X)=\inf\left\lbrace \int_{\mathcal{I}}\rho^i(X^i)d\mu\colon\int_{\mathcal{I}}X^id\mu=X\right\rbrace.\] 	However, it would be necessary to impose assumptions on $\mathcal{G}$ in order to avoid measurability issues. One of those assumptions is that the maps $i\rightarrow X^i(\omega)$ are measurable for any $\omega\in\Omega$ and every family $\{X^i\in L^\infty,\:i\in\mathcal{I}\}$, this would imply that $\mathcal{G}$ is the power set, which would leave us without meaningful choices for probability measures. In fact, as a consequence of Ulam's theorem, every probability measure on the power set of a set with cardinality as the power set of $\mathbb{N}$ is a discrete probability measure. In this sense, our approach can be considered the closest connection to the limiting integral case.
 
 
 The remainder of this paper is organized as follows.  In Section \ref{sec:prop}, we present the proposed approach and results regarding the preservation of financial and continuity properties from the set of risk measures. In Section \ref{sec:dual}, we prove results regarding dual representations for the convex, coherent, law-invariant, and comonotonic cases. In Section \ref{sec:alloc}, we explore optimal allocations by considering general results regarding the existence, comonotonic improvement, and law invariance of solutions, as well as the comonotonicity and flatness of distributions. In Section \ref{sec:special}, we explore the particular topic of self-convolution and its relation to regulatory arbitrage. In Section \ref{sec:exm} we expose concrete examples of our results on specific choices for families of risk measures.

 	\section{Proposed approach}\label{sec:prop}
 	
 		We consider a probability space $(\Omega,\mathcal{F},\mathbb{P})$. All equalities and inequalities are in the $\mathbb{P}-a.s.$ sense.  Let $L^0=L^0(\Omega,\mathcal{F},\mathbb{P})$ and $L^{\infty}=L^{\infty}(\Omega,\mathcal{F},\mathbb{P})$ be the spaces of (equivalence classes under $\mathbb{P}-a.s.$ equality of) finite and essentially bounded random variables, respectively. When not explicit, we consider in  $L^\infty$ its strong topology. We define $1_A$ as the indicator function for an event $A\in\mathcal{F}$. We identify constant random variables with real numbers. A pair $X,Y\in L^0$ is called comonotone if $\left( X(w)-X(w^{\prime})\right)\left( Y(w)-Y(w^{\prime}) \right)\geq0,\:\:w,w^{'}\in\Omega$ holds $\mathbb{P}\otimes\mathbb{P}-a.s.$ We denote by $X_n\rightarrow X$ convergence in the $L^\infty$ essential supremum norm $\lVert \cdot\rVert_{\infty}$, whereas $\lim\limits_{n\rightarrow\infty}X_n=X$ indicates $\mathbb{P}-a.s.$ convergence. The notation $X\succeq Y$, for $X,Y\in L^\infty$, indicates second-order stochastic dominance, that is, $E[f(X)]\leq E[f(Y)]$ for any increasing convex function $f\colon\mathbb{R}\rightarrow\mathbb{R}$. In particular, $E[X|\mathcal{F}^\prime]\succeq X$ for any $\sigma$-algebra $\mathcal{F}^\prime\subseteq\mathcal{F}$.

 	Let $\mathcal{P}$ be the set of all probability measures on $(\Omega,\mathcal{F})$. We denote, by $E_{\mathbb{Q}}[X]=\int_{\Omega}Xd\mathbb{Q}$, $F_{X, \mathbb{Q}}(x)=\mathbb{Q}(X\leq x)$, and $F_{X, \mathbb{Q}}^{-1}(\alpha)=\inf\left\lbrace x:F_{X, \mathbb{Q}}(x)\geq\alpha\right\rbrace $, the expected value, the (increasing and right-continuous) probability function, and its left quantile for $X\in L^\infty$ with respect to $\mathbb{Q}\in\mathcal{P}$. We write $X\overset{\mathbb{Q}}\sim Y$ when $F_{X,\mathbb{Q}}=F_{Y,\mathbb{Q}}$. We drop subscripts indicating probability measures when $\mathbb{Q}=\mathbb{P}$. Furthermore, let $\mathcal{Q}\subset\mathcal{P}$ be the set of probability measures $\mathbb{Q}$ that are absolutely continuous with respect to $\mathbb{P}$, with Radon--Nikodym derivative $\frac{d\mathbb{Q}}{d\mathbb{P}}$.  We denote the topological dual $(L^\infty)^*$ of $L^\infty$ by $ba$, which is defined as the space of finitely additive signed measures (with finite total variation norm $\lVert\cdot\rVert_{TV}$) that are absolutely continuous with respect to $\mathbb{P}$; moreover, we let $ba_{1,+}=\{m\in ba\colon m\geq0,m(\Omega)=1\}$ and by abuse of notation, we define $E_m[X]=\int_\Omega Xdm$ as the bilinear-form integral of $X\in L^\infty$ with respect to $m\in ba_{1,+}$.
 	
 	We begin with the definition of risk measures and the financial properties we consider in this paper. For more details regarding these properties, we refer to the classic books mentioned in the introduction.
 	
 	\begin{Def}\label{def:risk}
 		A functional $\rho:L^\infty\rightarrow\mathbb{R}$ is called a risk measure. It may have the following properties: 
 		
 		\begin{enumerate}
 			\item Monotonicity: If $X \leq Y$, then $\rho(X) \geq \rho(Y),\:\forall\: X,Y\in L^\infty$.
 			\item Translation invariance: $\rho(X+C)=\rho(X)-C,\:\forall\: X\in L^\infty,\:\forall\:C \in \mathbb{R}$.
 			\item Convexity: $\rho(\lambda X+(1-\lambda)Y)\leq \lambda \rho(X)+(1-\lambda)\rho(Y),\:\forall\: X,Y\in L^\infty,\:\forall\:\lambda\in[0,1]$.
 			\item Positive homogeneity: $\rho(\lambda X)=\lambda \rho(X),\:\forall\: X\in L^\infty,\:\forall\:\lambda \geq 0$.
 			\item Law invariance: If $F_X=F_Y$, then $\rho(X)=\rho(Y),\:\forall\:X,Y\in L^\infty$.
 			\item Comonotonic additivity: $\rho(X+Y)= \rho(X)+\rho(Y),\:\forall\: X,Y\in L^\infty$ with $X,Y$ comonotone.
 			\item Loadedness: $\rho(X)\geq -E[X],\:\forall\:X\in L^\infty$.
 			\item Limitedness: $\rho(X)\leq -\operatorname{ess}\inf X,\:\forall\:X\in L^\infty$. 
 		\end{enumerate}
 		
 		A risk measure  $\rho$ is called monetary if it satisfies (i) and (ii), convex if it is monetary and satisfies (iii), coherent if it is convex and satisfies (iv), law invariant if it satisfies (v), comonotone if it satisfies (vi), loaded if it satisfies (vii), and limited if it satisfies (viii). Unless otherwise stated, we assume that risk measures are normalized in the sense that $\rho(0)=0$. The acceptance set of $\rho$ is defined as $\mathcal{A}_\rho=\left\lbrace X\in L^\infty:\rho(X)\leq 0 \right\rbrace $.
 	\end{Def}

 	Let $\rho_\mathcal{I}=\{\rho^i\colon L^\infty\rightarrow\mathbb{R},\:i\in\mathcal{I}\}$ be some (a priori specified) collection of normalized monetary risk measures, where $\mathcal{I}$ is a nonempty infinite countable set. We define the set of weighting schemes $\mathcal{V}=\left\lbrace \{\mu_i\}_{i\in\mathcal{I}}\subset[0,1]\colon\:\sum_{i\in \mathcal{I}}\mu_i=1\right\rbrace $. Otherwise stated we fix $\mu\in\mathcal{V}$ and denote $\mathcal{I}_\mu=\{i\in\mathcal{I}\colon\mu_i>0\}$. We could consider $\mu\subset(0,1]$ and then $\mathcal{I}=\mathcal{I}_\mu$ without any harm for our results. We use the notations $\{X^i\in L^\infty,\:i\in\mathcal{I}\}=\{X^i,\:i\in\mathcal{I}\}=\{X^i\}_{\:i\in\mathcal{I}}=\{X^i\}$ for families indexed over $\mathcal{I}$; these families should be understood as generalizations of $n$-tuples. For any $X\in L^\infty$, we define its allocations as \[\mathbb{A}(X)=\left\lbrace\{X^i\}_{\:i\in\mathcal{I}}\colon\sum_{i\in\mathcal{I}}X^i\mu_i=X,\:\{X^i\}\:\text{is bounded}\right\rbrace .\] Evidently,  $\omega\rightarrow\sum_{i\in\mathcal{I}}X^i(\omega)\mu_i$ defines a random variable in $L^\infty$ for any $\{X^i\}_{i\in\mathcal{I}}\in\mathbb{A}(X),\:X\in L^\infty$. We note that the identity $\sum_{i\in \mathcal{I}}X^i\mu_i=X$ should then be understood in the $\mathbb{P}-a.s.$ sense.

 	The  countable case we study can be regarded as $\mathcal{I}=\mathbb{N}$ where the set of allocations $\mathbb{A}(X)$ consists of all sequences $\{X^i\}_{i\in\mathbb{N}}\subset L^\infty$ such that the associated sequence
 	$\sum_{i=1}^n \mu_iX^i$ converges to $X$ in the $\mathbb{P}-a.s.$ sense. Note that if $\sum_{i=1}^nX^i\mu_i=X$ for some $n\in\mathbb{N}$, then $\sum_{i=1}^{n+k}X^i\mu_i=X$ for any $k\in\mathbb{N}$ by taking $X^i=0$ for $i>n$. In particular $\{X^1,\dots,X^n,0,\dots\}\in\mathbb{A}(X)$. We have that $\mathbb{A}(X)\not=\emptyset$ for any $X$ because we can select $X^i=X,\:\forall\:i\in\mathcal{I}$. We also note that $\{X^i\}_{i\in\mathcal{I}}\in\mathbb{A}(X+Y)$ is equivalent to $\{X^i-Y\}_{i\in\mathcal{I}}\in\mathbb{A}(X)$ for any $X,Y\in L^\infty$. Furthermore, if  $\{X^i\}_{i\in\mathcal{I}}\in\mathbb{A}(X)$ and  $\{Y^i\}_{i\in\mathcal{I}}\in\mathbb{A}(Y)$, then $\{aX^i+bY^i\}\in\mathbb{A}(aX+bY)$ for any $a,b\in\mathbb{R}$ and $X,Y\in L^\infty$. We now define the core functional in our study.

 	\begin{Def}
 		Let $\rho_\mathcal{I}=\{\rho^i\colon L^\infty\rightarrow\mathbb{R},\:i\in\mathcal{I}\}$ be a collection of monetary risk measures and $\mu\in\mathcal{V}$. The $\mu$-weighted inf-convolution risk measure is a functional $\rho^\mu_{conv}\colon L^\infty\rightarrow\mathbb{R}\cup\{-\infty\}$ defined as
 		\begin{equation}\label{eq:conv}
 		\rho^\mu_{conv}(X)=\inf\left\lbrace\sum_{i\in \mathcal{I}}\rho^i(X^i)\mu_i\colon\{X^i\}_{i\in\mathcal{I}}\in\mathbb{A}(X)\right\rbrace. 
 		\end{equation}
 	\end{Def}

 	\begin{Rmk}
 	 We defined risk measures as functionals that only assume finite values. By abuse of notation, we will also consider $\rho^\mu_{conv}$ to be a risk measure, and we will provide conditions whereby it is finite. Since $\rho_{\mathcal{I}}$ consists of monetary risk measures, then $\rho^\mu_{conv}(X)<\infty$ because for any $X\in L^\infty$ we have that $\rho^\mu_{conv}(X)\leq\sum_{i\in\mathcal{I}}\rho^i(X)\mu_i\leq \lVert X\rVert_\infty<\infty$. Moreover, we note that normalization is not directly inherited from $\rho_\mathcal{I}$; indeed, $\rho^\mu_{conv}(0)\leq 0$. When $\rho^\mu_{conv}$ is convex it is finite if and only if $\rho^\mu_{conv}(0)>-\infty$, which is a well-known fact from convex analysis that a convex function that does not assume $\infty$ is either finite or $-\infty$ point-wise (see Lemma 16 of \cite{Delbaen2012} for instance). 

 	\end{Rmk}

 	The following proposition provides useful results regarding well definiteness,  representations, and properties of $\rho^\mu_{conv}$. 
 	
 	\begin{Prp}\label{prp:bound}
 		We have that
 		\begin{enumerate}
 		\item $\rho^\mu_{conv}$ is well defined.
 			\item For any $X\in L^\infty$ it holds that \begin{align*}
 		\rho^\mu_{conv}(X)&=\inf\left\lbrace \sum_{i\in\mathcal{I}}\rho^i(X-X^i)\mu_i\colon\{X^i\}_{i\in\mathcal{I}}\in\mathbb{A}(0)\right\rbrace\\
 		&=\lim\limits_{n\to\infty}\inf\left\lbrace \sum_{i=1}^n\rho^i(X^i)\mu_i\colon\sum_{i=1}^nX^i\mu_i=X \right\rbrace\\
 			&=\inf\left\lbrace \sum_{i=1}^n\rho^i(X^i)\mu_i\colon n\in\mathbb{N},\sum_{i=1}^nX^i\mu_i=X \right\rbrace.
 				\end{align*}
 			\item If $\rho_\mathcal{I}$ consists of risk measures satisfying positive homogeneity, then $\rho^\mu_{conv}(X)\leq\rho^i(X)\:\forall\:i\in \mathcal{I}_\mu,\:\forall\:X\in L^\infty$. 
 		\end{enumerate}  
 	\end{Prp}
 	\begin{proof}
 		\begin{enumerate}
 			\item We must to show that $\sum_{i\in \mathcal{I}}\rho^i(X^i)\mu_i $ converges for any $\{X^i\}\in\mathbb{A}(X)$.  As $\{X^i\}$ is bounded, there is $y,z \in \mathbb{R}$ such that $y  \leq X^i\leq z, \forall\:i\in\mathcal{I}$. By monotonicity and translation invariance it follows that $-z=\rho^i(z) \leq \rho^i(X^i) \leq \rho^i(y)=-y,\:\forall\:i\in\mathcal{I}$. Thus,
 			\[
 			-z=-z \sum_{i\in\mathcal{I}}\mu_i=- \sum_{i\in\mathcal{I}}z\mu_i\leq \sum_{i\in\mathcal{I}}\rho^i(X^i)\mu_i \leq -\sum_{i\in\mathcal{I}}y \mu_i = -y \sum_{i\in\mathcal{I}}\mu_i = -y.
 			\]
 		Hence, the claim follows by dominated convergence.
 			\item For the first relation, we note that $\sum_{i\in\mathcal{I}}X^i\mu_i=X$ if and only if $\sum_{i\in\mathcal{I}}(X-X^i)\mu_i=\sum_{i\in\mathcal{I}}(X^i-X)\mu_i=0$. Thus, by letting $Y^i=X-X^i,\:\forall\:i\in\mathcal{I}$, we have that  \[\rho^\mu_{conv}(X)=\inf\left\lbrace \sum_{i\in\mathcal{I}}\rho^i(X-Y^i)\mu_i\colon\{Y^i\}_{i\in\mathcal{I}}\in\mathbb{A}(0)\right\rbrace. \] Regarding the second relation,  
 			first notice that
 			\[
 			\lim\limits_{n\rightarrow\infty} \inf  \left\lbrace  \sum_{i=1}^n \rho(X^i)\mu_i : \sum_{i=1}^n X^i\mu_i = X \right\rbrace = \inf \bigcup_{n \in \mathbb{N}}   \left\lbrace  \sum_{i=1}^n \rho(X^i)\mu_i :  \sum_{i=1}^n X^i\mu_i = X \right\rbrace.
 			\]
 			Let $B_n :=  \left\lbrace  \sum_{i=1}^n \rho(X^i)\mu_i :   \sum_{i=1}^n X^i\mu_i = X \right\rbrace$ and $B :=   \left\lbrace  \sum_{i=1}^\infty  \rho(X^i)\mu_i :   \{X^i\}\in\mathbb{A}(X)\right\rbrace$. Such sets depend on $X$ and $B_n \subseteq B_{n+1} \subseteq B \subseteq \mathbb{R}, \forall \; n \in \mathbb{N}$. As any convergent infinite sum is the limit of finite sums, we obtain $B \subseteq cl (\cup_{n \in \mathbb{N}} B_n)$. Thus, $ cl(B) =  cl (\cup_{n \in \mathbb{N}} B_n)$. If both $B$ and $\cup_{n \in \mathbb{N}} B_n$ are unbounded from below then their infimum coincide to $-\infty$. If they both are bounded from below, we have that $ \rho^\mu_{conv} (X)  = \inf B = \inf cl(B) = \inf ( cl (\cup_{n \in \mathbb{N}} B_n)) = \inf  \cup_{n \in \mathbb{N}} B_n = \lim_n \inf \left\lbrace  \sum_{i=1}^n \rho(X^i)\mu_i : \sum_{i=1}^n X^i\mu_i = X \right\rbrace$. Therefore, we only need to show that $B$ is unbounded from below if and only if  $\cup_{n \in \mathbb{N}} B_n$ is unbounded from below. Since $\cup_{n \in \mathbb{N}} B_n \subseteq B$ we clearly have that if $\cup_{n \in \mathbb{N}} B_n$ is unbounded from below so is $B$. For the converse, let $B$ be unbounded from below. Then there is a sequence $\{b_j\} \subseteq B$ such that $b_j \downarrow - \infty$. As any $b_j$ is a limit point of a sequence (in $n$) $\{a_n^j\}\subseteq\cup_{n \in \mathbb{N}} B_n$, we can find another sequence (in $j$) $\{a^j_{n(j)}\} \subseteq \cup_{n \in \mathbb{N}} B_n$, where $n(j)$ is a sufficiently large natural number, such that $a_{n(j)}^j \rightarrow -\infty $. This fact implies that $\cup_{n \in \mathbb{N}} B_n$ is also unbounded from below.
 			
 			For the third relation, we show that $n\to\inf\left\lbrace \sum_{i=1}^n\rho^i(X^i)\mu_i\colon \sum_{i=1}^nX^i\mu_i=X \right\rbrace$ is decreasing. We have that \begin{align*}
 			&\inf\left\lbrace \sum_{i=1}^{n+1}\rho^i(X^i)\mu_i\colon \sum_{i=1}^{n+1}X^i\mu_i=X \right\rbrace\\
 			\leq&\inf\left\lbrace \sum_{i=1}^{n+1}\rho^i(X^i)\mu_i\colon \sum_{i=1}^nX^i\mu_i=X,X^{n+1}=0 \right\rbrace\\
 			=&\inf\left\lbrace \sum_{i=1}^n\rho^i(X^i)\mu_i\colon \sum_{i=1}^nX^i\mu_i=X \right\rbrace.
 			\end{align*}
 			Hence, the infimum with respect to $n\in\mathbb{N}$ can be replaced by a limit.
 			\item We assume, toward a contradiction, that there is $X\in L^\infty$ such that $\rho^\mu_{conv}(X)>\rho^j(X)$ for some $j\in\mathcal{I}_\mu$. Let $\{Y^i\}_{i\in\mathcal{I}}$ be such that $Y^i=(\mu_j)^{-1}X$ for $i=j$, and $Y^i=0$ otherwise. Then, $\{Y^i\}_{i\in\mathcal{I}}\in\mathbb{A}(X)$.  Thus, by positive homogeneity and the definition of $\rho^\mu_{conv}$  we have that \[\rho^j(X)<\rho^\mu_{conv}(X)\leq\sum_{i\in\mathcal{I}}\rho^i(Y^i)\mu_i=\rho^j(X),\] which is a contradiction. In this case, for any $X\in L^\infty$, we have that $\rho^\mu_{conv}(X)\leq\inf_{i\in\mathcal{I}_\mu}\rho^i(X)<\infty$.
 		\end{enumerate}
 		
 	\end{proof}
 	
 	\begin{Rmk}
 	We have that our convex approach is related to the countable summation risk sharing problem without weights. Let
 	\[\A^*(X) :=\left\lbrace\{X^i\}_{\:i\in\mathcal{I}}\colon\sum_{i\in\mathcal{I}}X^i=X,\:\{X^i\}\:\text{is bounded}\right\rbrace.\]
 	note that for each $\{ X^i \} \in \A(X) $ we have that $\{Y^i=\mu^iX^i\}\in\A^{*}(X)$. Similarly, if $\{ X^i \} \in \A^{*}(X) $ we have that $\left\lbrace Y^i=\frac{X^i}{\mu^i},\:i\in\mathcal{I}_\mu\right\rbrace \in\A(X)$. It is clear that there exist a one to one correspondence between $\A(X) $ and $\A^*(X)$ for all $X \in L^\infty$. Moreover, let
 	\[ \rho_\I^* := \left\lbrace \rho_*^i(X)=\p \left(\dfrac{X}{\mu_i}\right) \mu_i\:\forall\:X\in L^\infty, i\in \Iu \right\rbrace\]and\[\rho_{conv} (X) := \inf\left\lbrace\sum_{i\in \Iu}\rho_*^i(X^i) \colon\{X^i\}_{i\in\Iu}\in \A^*(X) \right\rbrace. \]
 	Note that any $\p_* \in \rho^*_\I$ inherits all relevant properties of $\p \in \rho_\I$ used in this study. Directly from those definitions we have that \begin{align*}
 	\rho_{conv}(X) &=   \inf\left\lbrace\sum_{i\in \Iu}\rho^i\left( \dfrac{X^i}{\mu_i}\right) \mu_i \colon\{ X^i\}\in \A^*(X) \right\rbrace
 	\\ &=  \inf\left\lbrace\sum_{i\in \I}\rho^i\left( X^i \right) \mu_i \colon\{ \mu_iX^i\}_{i\in\I}\in \A^{*}(X) \right\rbrace
 	\\ &=  \inf\left\lbrace\sum_{i\in \I}\rho^i\left( X^i \right) \mu_i \colon\{ X^i\}_{i\in\I}\in \A(X)  \right\rbrace
 	= \pcon (X).
 	\end{align*}
 	Of course, under positive homogeneity $\rho_\mathcal{I}=\rho^{*}_\mathcal{I}$ for any $i\in\mathcal{I}_\mu$ and, consequently, our approach becomes the additive risk sharing.
 	 	\end{Rmk}

  We now present a result regarding the preservation by $\rho^\mu_{conv}$ of financial properties of $\rho_{\mathcal{I}}$.
 	
 	\begin{Prp}\label{prp:propconv}
 		$\rho^\mu_{conv}$ is monetary. Moreover, if $\rho_\mathcal{I}$ consists of risk measures with  convexity, positive homogeneity, law invariance, loadedness, or limitedness, then each property is inherited by $\rho^\mu_{conv}$.
 	\end{Prp}
 	 	
 	\begin{proof}
 		Monotonicity, translation invariance,  convexity, positive homogeneity, loadedness and limitedness are directly obtained from definition of $\rho^\mu_{conv}$.  For law invariance, we begin by showing that $\rho^\mu_{conv}$ inherits law invariance on the sub-domain \[L^\infty_{\perp}:=\{X\in L^\infty\colon\exists\:\text{uniform on}\:[0,1]\:\text{r.v. independent of}\:X\}.\] To that, let $X,Y\in L^\infty_{\perp}$ with $X\sim Y$ and take some arbitrary $\{X^i\}_{i\in\mathcal{I}}\in\mathbb{A}(X)$. Note that we can have a countable set $\{U^i,\:i\in\mathcal{I}\}$ of i.i.d. uniform on $[0,1]$ random variables independent of $Y$ because our probability space is atomless, see Theorem 1 of \cite{Delbaen2012} for instance. Now take $X^0=X$, $Y^0=Y$ and let $Y^i=F^{-1}_{X^i|X^{i-1},\cdots,X^0}(U^i|Y^{i-1},\cdots,Y^0)\:\forall\:i\in\mathcal{I}$, which is the conditional quantile function. We thus get that $(Y,Y^1,\cdots,Y^n)\sim (X,X^1,\cdots,X^n)\:\forall\:n\in\mathcal{I}$ and $\{Y^i\}_{i\in\mathcal{I}}\in\mathbb{A}(Y)$. In this sense we obtain $\rho^\mu_{conv}(Y)\leq\sum_{i\in\mathcal{I}}\rho^i(Y^i)\mu_i=\sum_{i\in\mathcal{I}}\rho^i(X^i)\mu_i$. Taking the infimum over $\mathbb{A}(X)$ we get $\rho^\mu_{conv}(Y)\leq\rho^\mu_{conv}(X)$. By reversing roles of $X$ and $Y$, we obtain Law Invariance on $L^\infty_{\perp}$. Now, let $X,Y\in L^\infty$ with $X\sim Y$ and define $\{X_n\}\subset L^\infty$ as $X_n=\frac{1}{n}\lfloor nX\rfloor$ and $\{Y_n\}\subset L^\infty$ as $Y_n=\frac{1}{n}\lfloor nY\rfloor$, where $\lfloor\cdot\rfloor$ is the floor function. Moreover, it is easy to show that $X_n\sim Y_n,\:\forall\:n\in\mathbb{N}$. By Lemma 3 in \cite{Liu2019} we have that if $X\in L^\infty$ takes values in a countable set, then $X\in L^\infty_{\perp}$. Thus $\{X_n\}\subset L^\infty_\perp$ and $\rho^\mu_{conv}(X_n)=\rho^\mu_{conv}(Y_n),\:\forall\:n\in\mathbb{N}$. We have that $\rho^\mu_{conv}$ is Lipschitz continuous since it is monetary. In particular, it possesses continuity in $\lVert\cdot\rVert_\infty$ norm. Note that $X_{2^n}\rightarrow X$. Thus $|\rho^\mu_{conv}(X)-\rho^\mu_{conv}(Y)|\leq\lim\limits_{n\rightarrow\infty}|\rho^\mu_{conv}(X)-\rho^\mu_{conv}(X_{2^n})|+\lim\limits_{n\rightarrow\infty}|\rho^\mu_{conv}(Y)-\rho^\mu_{conv}(Y_{2^n})|=0$. Hence $\rho^\mu_{conv}(X)=\rho^\mu_{conv}(Y)$.
 	\end{proof}
 	
 	\begin{Rmk}\label{rmk:prop}
 		\begin{enumerate}
 			\item Concerning the preservation of subadditivity, that is, $\rho(X+Y)\leq\rho(X)+\rho(Y)$, the result follows by an argument analogous to that for convexity, but with $X+Y$ instead of $\lambda X+(1-\lambda)Y$. We note that in this case, we have normalization because $\rho^\mu_{conv}(0)\leq0$, whereas $\rho^\mu_{conv}(X)\leq\rho^\mu_{conv}(X)+\rho^\mu_{conv}(0)$, which implies $\rho^\mu_{conv}(0)\geq0$. If the risk measures of $\rho_{\mathcal{I}}$ are loaded, we also have normalization because $0=\rho(0)\geq\rho^\mu_{conv}(0)\geq E[-0]=0$. Of course, in the case of positive homogeneity, we also obtain normalization. 
 			\item	Regarding the preservation of comonotonic additivity, let $X,Y\in L^\infty$ be a comonotone pair. Then, $\lambda X$,$(1-\lambda)Y$ is also comonotone for any $\lambda\in[0,1]$. We note that for any monetary risk measure $\rho$, comonotonic additivity implies positive homogeneity. Then, we have that \[\rho^\mu_{conv}(X+Y)=2\rho^\mu_{conv}\left(\frac{X}{2}+\frac{Y}{2} \right)\leq2\left(\frac{1}{2}\rho^\mu_{conv}(X)+\frac{1}{2}\rho^\mu_{conv}(Y) \right) =\rho^\mu_{conv}(X)+\rho^\mu_{conv}(Y).\] Thus, we obtain subadditivity for comonotone pairs. If we have, additionally, convexity (and hence coherence) for $\rho_{\mathcal{I}}$, then comonotonic additivity is preserved, as shown in Theorem \ref{Thm:dualLI}.
 		\end{enumerate}
 		
 	\end{Rmk}

 In addition to the usual norm-based continuity notions, $\mathbb{P}-a.s.$ pointwise continuity notions are relevant in the context of risk measures. In the following, we focus on the preservation by $\rho^\mu_{conv}$ of continuity properties of $\rho_{\mathcal{I}}$
 	
 	\begin{Def}\label{def:cont}
 		A risk measure $\rho:L^\infty\rightarrow\mathbb{R}$ is called
 		\begin{enumerate}
 			\item Fatou continuous:  If $\lim\limits_{n\rightarrow\infty}X_n=X$ implies that $\rho(X) \leq \liminf\limits_{n\rightarrow\infty} \rho( X_{n})$, $\forall\:\{X_n\}_{n=1}^\infty$ bounded in $L^\infty$ norm and for any $X\in L^\infty$.
 			\item Continuous from above: If $\lim\limits_{n\rightarrow\infty}X_n=X$, with $\{X_n\}$ being decreasing, implies that $\rho(X)= \lim\limits_{n\rightarrow\infty} \rho( X_{n})$, $\:\forall\:\{X_n\}_{n=1}^\infty,X\in L^\infty$.
 			\item Continuous from below: If $\lim\limits_{n\rightarrow\infty}X_n=X$, with $\{X_n\}$ being increasing, implies that $\rho(X)= \lim\limits_{n\rightarrow\infty} \rho( X_{n})$, $\:\forall\:\{X_n\}_{n=1}^\infty,X\in L^\infty$.
 			\item Lebesgue continuous: If $\lim\limits_{n\rightarrow\infty}X_n=X$ implies that $\rho(X)= \lim\limits_{n\rightarrow\infty} \rho( X_{n})$, $\:\forall\:\{X_n\}_{n=1}^\infty$ bounded  in $L^\infty$ norm and $X\in L^\infty$.
 		\end{enumerate}
 	\end{Def}
 	
 	\begin{Prp}\label{prp:propconv2}
 		 We have that
 		\begin{enumerate}
 			\item  If $\rho_\mathcal{I}$ consists of Lipschitz continuous risk measures in relation to some pseudo metric $d$, then $\rho^\mu_{conv}$ is  Lipschitz continuous in that pseudo-metric.
 			\item If $\rho_\mathcal{I}$ consists of continuous from below risk measures, then $\rho^\mu_{conv}$ is continuous from below.
 		\end{enumerate}
 	\end{Prp}
 	
 	\begin{proof}
 		\begin{enumerate}
 			\item For each $i\in\mathcal{I}$, we have that $|\rho^i(X)-\rho^i(Y)|\leq C d( X,Y),\:C>0$.  Thus,
 			\begin{align*}
 			&\left|\rho^\mu_{conv}(X)-\rho^\mu_{conv}(Y)\right|\\
 			=&\left| \inf\left\lbrace \sum_{i\in\mathcal{I}}\rho^i(X-X^i)\mu_i\colon\{X^i\}_{i\in\mathcal{I}}\in\mathbb{A}(0)\right\rbrace-\inf\left\lbrace \sum_{i\in\mathcal{I}}\rho^i(Y-X^i)\mu_i\colon\{X^i\}_{i\in\mathcal{I}}\in\mathbb{A}(0)\right\rbrace\right|\\
 			\leq&\sup\left\lbrace\left|\sum_{i\in\mathcal{I}}\left[\rho^i(X-X^i)-\rho^i(Y-X^i) \right]\mu_i  \right|\colon\{X^i\}_{i\in\mathcal{I}}\in\mathbb{A}(0)  \right\rbrace\leq C d( X,Y).
 			\end{align*}
 			\item 	Let $\{X_n\}_{n=1}^\infty\subset L^\infty$ be increasing such that $\lim\limits_{n\rightarrow\infty}X_n=X\in L^\infty$. By the monotonicity of $\rho_{\mathcal{I}}$ we have that each $\rho^i(X_n-X^i) $ is decreasing in $n$. Moreover, $i\to \rho^i(X_n-X^i)$ is bounded above by $\sup_{i\in\mathcal{I}}\lVert X_1-X^i\rVert_\infty<\infty$. Thus, by the monotone convergence Theorem we have that
 			\begin{align*}
 			\lim\limits_{n\rightarrow\infty}\rho^\mu_{conv}(X_n)&=\inf\limits_{n}\left\lbrace \inf\limits_{\{X^i\}\in\mathbb{A}(0)} \sum_{i\in\mathcal{I}}\rho^i(X_n-X^i)\mu_i\right\rbrace\\
 			&=\inf\limits_{\{X^i\}\in\mathbb{A}(0)}\left\lbrace \inf\limits_{n}\sum_{i\in\mathcal{I}}\rho^i(X_n-X^i)\mu_i\right\rbrace \\
 			&=\inf\limits_{\{X^i\}\in\mathbb{A}(0)}\left\lbrace \sum_{i\in\mathcal{I}}\left[\inf\limits_{n}\rho^i(X_n-X^i)\right] \mu_i\right\rbrace =\rho^\mu_{conv}(X).
 			\end{align*}
 		\end{enumerate}
 	\end{proof}
 	
 	\begin{Rmk}\label{rmk:prop2}
 		\begin{enumerate}
 			\item It is important to note that Fatou continuity is not preserved even when $\mathcal{I}$ is finite, as $\lim\limits_{n\rightarrow\infty}X_n=X$ does not imply the existence of $\{X^i_n\}_{i\in\mathcal{I}}\in\mathbb{A}(X_n)\:\forall\:n\in\mathbb{N}$ and $\{X^i\}_{i\in\mathcal{I}}\in\mathbb{A}(X)$ such that $\lim\limits_{n\rightarrow\infty}X^i_n=X^i,\:\forall\:i\in\mathcal{I}$. See Example 9 in \cite{Delbaen2000}, for instance. Accordingly, one should be careful when dual representations that depends of such continuity property are considered. 
 		\item 	 If $\rho_{\mathcal{I}}$ consists of convex risk measures that are continuous from below, then $\rho^\mu_{conv}$ is convex Lebesgue continuous. This is true because continuity from below is equivalent to Lebesgue continuity for convex risk measures. See Theorem 4.22 in \cite{Follmer2016} for instance.
 	\item 		 Robustness is a key concept in the presence of model uncertainty. It implies a small variation in the output functional when there is bad specification. See, for instance, \cite{Cont2010}, \cite{Kratschmer2014}, and \cite{Kiesel2016}. Formally, if $d$ is a pseudo-metric on $L^\infty$, then a risk measure $\rho\colon L^\infty\rightarrow\mathbb{R}$ is called $d$-robust if it is continuous with respect to $d$.	 In light of Proposition \ref{prp:propconv2}, we have that the continuity of risk measures in $\rho_{\mathcal{I}}$ with respect to $d$ is not generally preserved by $\rho^\mu_{conv}$. Consequently, the same is true for $d$-robustness. Nonetheless, under Lipschitz continuity, we have the preservation of robustness.
 \end{enumerate}
 		\end{Rmk}



 	\section{Dual representations}\label{sec:dual}

 	We now present the main results regarding the representation of $\rho^\mu_{conv}$ for convex cases. To that, we need the following fundamental result.
 	
 	 	\begin{Thm}[Theorem 2.3 in \cite{Delbaen2002}, Theorem 4.33 in \cite{Follmer2016}]\label{the:dual}
 		Let $\rho : L^\infty\rightarrow \mathbb{R}$ be a risk measure. Then,
 		\begin{enumerate}
 			\item $\rho$ is a  convex risk measure if and only if it can be represented as
 			\begin{equation}\label{eq:dual}
 			\rho(X)=\max\limits_{m\in ba_{1,+}}\left\lbrace E_m[-X]-\alpha^{min}_\rho(m) \right\rbrace,\:\forall\:X\in L^\infty,
 			\end{equation}
 			where  $\alpha^{min}_\rho : ba_{1,+}\rightarrow\mathbb{R}_+\cup\{\infty\}$, defined as $\alpha^{min}_\rho(m)=\sup\limits_{X\in L^\infty}\left\lbrace E_m[-X]-\rho(X) \right\rbrace= \sup\limits_{X\in\mathcal{A}_\rho}E_m[-X]$, is a lower semi-continuous (in the total-variation norm) convex function that is called penalty term. 
 			\item 	$\rho$ is a coherent risk measure if and only if it can be represented as
 			\begin{equation}\label{eq:cohdual}
 			\rho(X)=\max\limits_{m\in\mathcal{Q}_\rho} E_m[-X],\:\forall\:X\in L^\infty,
 			\end{equation} where $\mathcal{Q}_\rho\subseteq ba_{1,+}$ is a nonempty, closed, and convex set that is called the dual set of $\rho$.
 		\end{enumerate}
 	\end{Thm}
 	
 	\begin{Rmk}\label{rmk:dual}
 		With the assumption of Fatou continuity, the representations in the previous theorem could be considered over $\mathcal{Q}$ instead of $ba_{1,+}$, but with the supremum not necessarily being attained. Moreover, for convex risk measures, we can define certain subgradients using Legendre--Fenchel duality (i.e., convex conjugates), as follows:
 		\begin{align*}
 		\partial \rho(X)&=\left\lbrace m\in ba_{1,+}\colon \rho(Y)-\rho(X)\geq E_m[-(Y-X)]\:\forall\:Y\in L^\infty\right\rbrace\\
 		&=\left\lbrace m\in ba_{1,+}\colon E_m[-X]-\alpha^{min}_\rho(m)\geq\rho(X)\right\rbrace,\\
 		&\\
 		\partial \alpha^{min}_\rho(m)&=\left\lbrace X\in L^\infty\colon \alpha^{min}_\rho(n)-\alpha^{min}_\rho(m)\geq E_{(n-m)}[-X]\:\forall\:n\in ba\right\rbrace\\
 		&=\left\lbrace X\in L^\infty\colon E_m[-X]-\rho(X)\geq\alpha^{min}_\rho(m)\right\rbrace,
 		\end{align*}
 		The negative sign in the expectation above is used to maintain the (anti) monotonicity pattern of risk measures. We note that these subgradient sets could be empty if we consider only $\mathcal{Q}$ instead of $ba_{1,+}$. Moreover, by Theorem \ref{the:dual}, we could replace the inequalities in the definition of sub-gradients by equalities. Further, it is immediate that $X\in\partial \alpha^{min}_\rho(m)$ if and only if $m\in \partial \rho(X)$. 
 	\end{Rmk}
 	\begin{Thm}\label{Thm:dualconv}
 		Let $\rho_\mathcal{I}$ be a collection of convex risk measures. We have that
 		\begin{enumerate}
 			\item The acceptance set of $\rho^\mu_{conv}$ is \begin{equation}\mathcal{A}_{\rho^{\mu}_{conv}}= cl(\mathcal{A}_\mu),\end{equation} where $\mathcal{A}_\mu=\left\lbrace X\in L^\infty\colon \exists\:\{X^i\}_{i\in\mathcal{I}}\in\mathbb{A}(X)\:s.t.\:X^i\in\mathcal{A}_{\rho^i}\:\forall\:i\in\mathcal{I}_\mu\right\rbrace=\sum_{i\in \mathcal{I}}\mathcal{A}_{\rho^i}\mu_i$. Moreover, $\mathcal{A}_\mu$ is not dense in $L^\infty$ if and only if $\mathcal{A}_\mu\not=L^\infty$.
 			\item The minimal penalty term of $\rho^\mu_{conv}$ is
 			\begin{equation}\label{eq:penalty}
 			\alpha^{min}_{\rho^{\mu}_{conv}}(m)=\sum_{i\in\mathcal{I}}\alpha^{min}_{\rho^i}(m)\mu_i,\:\forall\:m\in ba_{1,+}.
 			\end{equation} 
 			\end{enumerate}
 	\end{Thm}
 	
 	\begin{proof}
 We have that $\rho^\mu_{conv}$ is a convex risk measure that is either finite or identically $-\infty$. Moreover, it is well known for any $n\in\mathbb{N}$ that \[\sum_{i=1}^n\alpha^{min}_{\rho^i}\mu_i=\sup\limits_{X\in L^\infty}\left\lbrace E_m[-X]-\inf\left\lbrace \sum_{i=1}^n\rho^i(X^i)\mu_i\colon\sum_{i=1}^n X^i\mu_i=X  \right\rbrace \right\rbrace,\] which generates the acceptance set $ cl\left(\sum_{i=1}^ n\mathcal{A}_{\rho^i}\mu_i\right)$ with $\{\mu_1,\dots,\mu_n\}\subset[0,1]^n$. We now demonstrate the claims.
 		\begin{enumerate}
 	\item Let $X\in\mathcal{A}_\mu$. Then, there is $\{X^i\}_{i\in\mathcal{I}}\in\mathbb{A}(X)$ such that $X^i\in\mathcal{A}_{\rho^i}\:\forall\:i\in\mathcal{I}_\mu$. Thus, $\rho^\mu_{conv}(X)\leq\sum_{i\in \mathcal{I}}\rho^i(X^i)\mu_i\leq 0$. This implies $X\in\mathcal{A}_{\rho^{\mu}_{conv}}$. By taking closures we get $cl(\mathcal{A}_\mu)\subseteq\mathcal{A}_{\rho^{\mu}_{conv}}$. We note that $\mathcal{A}_\mu$ is not necessarily closed (for reasons similar to those for which $\rho^\mu_{conv}$ does not inherit Fatou continuity). For the converse relation, let $X\in int(\mathcal{A}_{\rho^\mu_{conv}})$. Then there is $\{X^i\}\in\mathbb{A}(X)$ such that $k=\sum_{i\in \mathcal{I}}\rho^i(X^i)\mu_i<0$. Since $\rho^i(X^i+\rho^i(X^i)-k)<0$ for any $i\in\mathcal{I}$, we have that $Y^i=X^i+\rho^i(X^i)-k\in int(\mathcal{A}_{\rho^i})$ for any $i\in\mathcal{I}$. Moreover, $\sum_{i\in \mathcal{I}}Y^i\mu_i=X+k-k=X$. Then $\{Y^i\}\in\mathbb{A}(X)$, which implies $X\in \mathcal{A}_\mu\subseteq cl(\mathcal{A}_\mu)$. Thus, $int(\mathcal{A}_{\rho^\mu_{conv}}) \subseteq \mathcal{A}_\mu$. By taking closures we get $cl(int(\mathcal{A}_{\rho^\mu_{conv}})) =  \mathcal{A}_{\rho^\mu_{conv}} \subseteq cl(\mathcal{A}_\mu)$, which gives the required equality. Moreover, let $\mathcal{A}_\mu$ be norm-dense in $L^\infty$. Thus, for any $X\in L^\infty$ and $k>0$, there is $Y\in\mathcal{A}_\mu$ such that $\lVert X-Y\rVert_\infty\leq k$. Thus, $X+k\geq Y$ and, as $\mathcal{A}_\mu$ is monotone, we obtain $X+k\in\mathcal{A}_\mu$. As both $X$ and $k$ are arbitrary, $\mathcal{A}_\mu=L^\infty$. The converse relation is trivial.
 			\item We have for any $m\in ba_{1,+}$ that 	
 		\begin{align*}
 	\alpha_{\rho^\mu_{conv}}^{min}(m)&= \sup\limits_{X\in L^\infty}\left\lbrace E_m[-X]-\lim\limits_{n\rightarrow\infty}\inf\left\lbrace \sum_{i=1}^n\rho^i(X^i)\mu_i\colon\sum_{i=1}^n X^i\mu_i=X  \right\rbrace \right\rbrace\\
 	&\leq \lim\limits_{n\rightarrow\infty} \sup\limits_{X\in L^\infty}\left\lbrace E_m[-X]-\inf\left\lbrace \sum_{i=1}^n\rho^i(X^i)\mu_i\colon\sum_{i=1}^n X^i\mu_i=X  \right\rbrace \right\rbrace\\
 &=\lim\limits_{n\rightarrow\infty}\sum_{i=1}^n\alpha_{\rho^i}^{min}(m)\mu_i= \sum_{i\in\mathcal{I}}\alpha_{\rho^i}^{min}(m)\mu_i.		\end{align*}

For the converse, by Theorem \ref{the:dual} and Remark \ref{rmk:dual}, we have that \[\alpha^{min}_{\rho^{\mu}_{conv}}(m)=\sup\limits_{X\in\mathcal{A}_{\rho^{\mu}_{conv}}}E_m[-X],\:\mathcal{A}_{\rho^{\mu}_{conv}}=\left\lbrace X\in L^\infty\colon\alpha^{min}_{\rho^{\mu}_{conv}}(X)\geq E_m[-X]\:\forall\:m\in ba_{1,+}\right\rbrace.\] 	
Moreover, note that if $X\in \sum_{i=1}^n\mathcal{A}_{\rho^i}\mu_i$, then $X=\sum_{i=1}^n X^i\mu_i$ with $X^i\in\mathcal{A}_{\rho^i}$ for $i=1,\dots,n$. Thus, $X\in\mathcal{A}_\mu$ since $\{X^1,\dots,X^n,0,0,\dots\}\in\mathbb{A}(X)$ and $0$ is acceptable for any $\rho^i$. Then, for any $m\in ba_{1,+}$ we have

\begin{align*}
\sum_{i\in\mathcal{I}}\alpha_{\rho^i}^{min}(m)\mu_i
&=\lim\limits_{n\rightarrow\infty}\sum_{i=1}^n\sup \left\lbrace  E_m[-X]\mu_i \colon {X\in \mathcal{A}_{\rho^i}} \right\rbrace \\
	&= \lim\limits_{n\rightarrow\infty}\sup \left\lbrace \sum_{i=1}^n E_m[-X^i]\mu_i  \colon  {\{X^i\in \mathcal{A}_{\rho^i}\}_{i=1,\dots,n}} \right\rbrace  \\
	&=\lim\limits_{n\rightarrow\infty}\sup  \left\lbrace  E_m[-X]  \colon {X\in \sum_{i=1}^n\mathcal{A}_{\rho^i}\mu_i}  \right\rbrace \\
	&\leq \lim\limits_{n\rightarrow\infty}\sup  \left\lbrace E_m[-X] \colon  {X\in \mathcal{A}_{\mu}}\right\rbrace 
	\\ &= \sup \left\lbrace E_m[-X] \colon  {X\in \mathcal{A}_{\rho^\mu_{conv}}}      \right\rbrace=\alpha_{\rho^\mu_{conv}}^{min}(m).
		\end{align*}

	Hence, $\alpha^{min}_{\rho^{\mu}_{conv}}=\sum_{i\in\mathcal{I}}\alpha^{min}_{\rho^i}\mu_i$. Regarding the properties of $m\rightarrow\sum_{i\in\mathcal{I}}\alpha^{min}_{\rho^i}(m)\mu_i$, non-negativity is straightforward, whereas convexity follows from the monotonicity of the integral and the convexity of each $\alpha^{min}_{\rho_i}$ because for any $\lambda\in[0,1]$ and  $m_1,m_2\in ba_{1,+}$, we have that \begin{align*}
 			\sum_{i\in\mathcal{I}}\alpha^{min}_{\rho^i}(\lambda m_1+(1-\lambda)m_2)\mu_i&\leq\sum_{i\in\mathcal{I}}\left[\lambda\alpha^{min}_{\rho^i}(m_1)+(1-\lambda)\alpha^{min}_{\rho^i}(m_2) \right]\mu_i\\
 			&=\lambda\sum_{i\in\mathcal{I}}\alpha^{min}_{\rho^i}( m_1)\mu_i+(1-\lambda)\sum_{i\in\mathcal{I}}\alpha^{min}_{\rho^i}(m_2)\mu_i.
 			\end{align*}
 			Furthermore, by Fatou's lemma (which can be used because each $\alpha^{min}_{\rho^i}$ is bounded from below by $0$) and by the lower semi-continuity of each $\alpha^{min}_{\rho^i}$ with respect to the total variation norm on $ba$, for any $\{m_n\}$ such that $m_n\rightarrow m$, we have that
 			\[\sum_{i\in\mathcal{I}}\alpha^{min}_{\rho^i}(m)\mu_i\leq\sum_{i\in\mathcal{I}}\liminf\limits_{n\rightarrow\infty}\alpha^{min}_{\rho^i}(m_n)\mu_i\leq\liminf\limits_{n\rightarrow\infty}\sum_{i\in\mathcal{I}}\alpha^{min}_{\rho^i}(m_n)\mu_i.\]

 		\end{enumerate}
 	\end{proof}

 	\begin{Rmk}\label{rmk:fatou}
 		\begin{enumerate}
 		\item Under the assumption of Fatou continuity for both the risk measures in $\rho_{\mathcal{I}}$ and $\rho^\mu_{conv}$, the claims in Theorem \ref{Thm:dualconv} could be adapted by replacing the finitely additive measures $m\in ba_{+,1}$ by probabilities $\mathbb{Q}\in\mathcal{Q}$. Moreover, weak$^*$ topological concepts could replace the corresponding strong (norm) topological concepts.
 			\item The weighted risk measure is a functional $\rho^\mu:L^\infty\rightarrow\mathbb{R}$ defined as
 			$\rho^\mu(X)=\sum_{i\in\mathcal{I}}\rho^{i}(X)\mu_i$. Theorem 4.6 in \cite{Righi2019} states that, assuming Fatou continuity, $\rho^\mu$ can be represented using a convex (not necessarily minimal) penalty  defined as 
 			\[\alpha_{\rho^\mu}(\mathbb{Q})=
 			\inf\left\lbrace \sum_{i\in\mathcal{I}}\alpha^{min}_{\rho^i}\left(\mathbb{Q}^i\right)\mu_i\colon \sum_{i\in\mathcal{I}}\mathbb{Q}^i\mu_i=\mathbb{Q},\:\mathbb{Q}^i\in\mathcal{Q}\:\forall\:i\in\mathcal{I}\right\rbrace.\] By the duality of convex conjugates, $\alpha_{\rho^\mu}=\alpha_{\rho^\mu}^{min}$ if and only if $\alpha_{\rho^\mu}$ is lower semi-continuous.  Nonetheless, this represents a connection between weighted and inf-convolution functions for countable $\mathcal{I}$, as in the traditional finite case.
 		\end{enumerate}
 		
 	\end{Rmk}


 	The following corollary provides interesting properties regarding the normalization, finiteness, preservation, dominance, and sub-gradients of $\rho^\mu_{conv}$. 
 
 	\begin{Crl}\label{crl:hull}
 		Let $\rho_\mathcal{I}$ be a collection of  convex risk measures. Then
 		\begin{enumerate}
 			\item $\left\lbrace m\in ba_{1,+}\colon\sum_{i\in\mathcal{I}}\alpha^{min}_{\rho^i}(m)\mu_i<\infty\right\rbrace \not=\emptyset$ if and only if  $\rho^{\mu}_{conv}$ is finite. In this case $A_\mu$ is not dense in $L^\infty$.
 			\item $\rho^{\mu}_{conv}$ is normalized if and only if $\left\lbrace m\in ba_{1,+} \colon\alpha^{min}_{\rho^i}(m)=0\:\forall\:i\in\mathcal{I}_\mu\right\rbrace \not=\emptyset$.
 			\item If $\rho\colon L^\infty\rightarrow\mathbb{R}$ is a convex risk measure with $\rho(X)\leq (\geq\:\text{or}\:=)\rho^i(X)\:\forall\:i\in\mathcal{I}_\mu,\:\forall\:X\in L^\infty$, then $\rho(X)\leq(\geq\:\text{or}\:=)\rho^\mu_{conv}(X),\:\forall\:X\in L^\infty$.
 			\item  $\left\lbrace m\in ba_{1,+}\colon m\in\partial\rho^i(X^i)\:\forall\:i\in\mathcal{I}_\mu \right\rbrace=\bigcap_{i\in\mathcal{I}_\mu}\partial\rho^i(X^i)\subseteq\partial\rho^\mu_{conv}(X)$ for any $\{X^i\}_{i\in\mathcal{I}}\in\mathbb{A}(X)$ and $X\in L^\infty$.
 		\item  $\left\lbrace X\in L^\infty\colon \exists\:\{X^i\}_{i\in\mathcal{I}}\in\mathbb{A}(X)\:s.t.\:X^i\in\partial\alpha_{\rho^i}^{min}(m)\:\forall\:i\in\mathcal{I}_\mu\right\rbrace=\sum_{i\in \mathcal{I}}\partial\alpha_{\rho^i}^{min}(m)\mu_i\subseteq\partial\alpha_{\rho^\mu_{conv}}^{min}(m)$ for any $m\in ba_{1,+}$.
 		 		\end{enumerate}
 	\end{Crl}
 	\begin{proof}
 		\begin{enumerate}
 			\item 
 			By Proposition \ref{prp:propconv}, we have $\rho^\mu_{conv}<\infty$. If $\left\lbrace m\in ba_{1,+}\colon\sum_{i\in\mathcal{I}}\alpha^{min}_{\rho^i}(m)\mu_i<\infty\right\rbrace \not=\emptyset$, then there exists $m\in ba_{1,+}$ such that \[-\infty<E_m[-X]-\sum_{i\in\mathcal{I}}\alpha^{min}_{\rho^i}(m)\mu_i\leq E_m[-X]-\alpha^{min}_{\rho^\mu_{conv}}(m)\leq \rho^\mu_{conv}(X).\] 
		If $\left\lbrace m\in ba_{1,+}\colon\sum_{i\in\mathcal{I}}\alpha^{min}_{\rho^i}(m)\mu_i<\infty\right\rbrace =\emptyset$, then $\alpha^{min}_{\rho^\mu_{conv}}(m)=\sum_{i\in\mathcal{I}}\alpha^{min}_{\rho^i}(m)\mu_i=\infty,\:\forall\:m\in ba_{1,+}$. Hence, $\rho^\mu_{conv}(X)=-\infty,\:\forall\:X\in L^\infty$. Furthermore, if $\mathcal{A}_{\mu}$ is dense in $L^\infty$, then by (i) in Theorem \ref{Thm:dualconv}, we have $\mathcal{A}_{\mu}=L^\infty$. Thus, we have $\rho^{\mu}_{conv}(X)=\inf\{m\in\mathbb{R}\colon X+m\in L^\infty\}=-\infty,\:\forall\:X\in L^\infty$.
 			\item Let $\left\lbrace m\in ba_{1,+}\colon\alpha^{min}_{\rho^i}(m)=0\:\forall\:i\in\mathcal{I}_\mu\right\rbrace \not=\emptyset$ and $m^\prime$ in this set. Then, $\alpha^{min}_{\rho^{\mu}_{conv}}(m^\prime)= \sum_{i\in\mathcal{I}}\alpha^{min}_{\rho^i}(m^\prime)\mu_i=0$. Hence, 
 			 \[\rho^{\mu}_{conv}(0)=-\min_{m\in ba_{1,+}}\alpha^{min}_{\rho^{\mu}_{conv}}(m)=0.\] For the converse relation, let  $\left\lbrace m\in ba_{1,+}\colon\alpha^{min}_{\rho^i}(m)=0\:\forall\:i\in\mathcal{I}_\mu\right\rbrace =\emptyset$. Thus, for any $m\in ba_{1,+}$, we have that $ \left\lbrace i\in\mathcal{I}_\mu\colon \alpha_{\rho^i}^{min}(m)>0 \right\rbrace \not=\emptyset$. Then, \[\rho^\mu_{conv}(0)=-\min_{m\in ba_{1,+}}\alpha^{min}_{\rho^\mu_{conv}}(m)<0,\] which implies that $\rho^\mu_{conv}$ is not normalized. 
 			\item We prove for the $\leq$ relation since the others are quite similar. By Theorem \ref{the:dual} and and Remark \ref{rmk:dual}, we have that $\alpha^{min}_{\rho}(m)\geq\alpha^{min}_{\rho^i}(m)\:\forall\:i\in\mathcal{I}_\mu$  for any $m\in ba_{1,+}$. Thus, by Theorem \ref{Thm:dualconv}, we obtain that \[\alpha^{min}_{\rho}(m)\geq\sum_{i\in\mathcal{I}}\alpha^{min}_{\rho^i}(m)\mu_i=\alpha^{min}_{\rho^\mu_{conv}}(m).\] Hence, $\rho(X)\leq\rho^\mu_{conv}(X),\:\forall\:X\in L^\infty$. 
 			
 	\item  Let $m^\prime\in\left\lbrace m\in ba_{1,+}\colon m\in\partial\rho^i(X^i)\:\forall\:i\in\mathcal{I}_\mu \right\rbrace$. Then, \[E_{m^\prime}[-X]-\alpha^{min}_{\rho^\mu_{conv}}(m^\prime)\geq\sum_{i\in\mathcal{I}}\left(E_{m^\prime}[-X^i]-\alpha^{min}_{\rho^i}(m^\prime)\right)\mu_i\geq\sum_{i\in\mathcal{I}}\rho^i(X^i)\mu_i\geq\rho^\mu_{conv}(X),\] which implies $m^\prime\in\partial\rho^\mu_{conv}(X)$. 
 	\item We fix $m\in ba_{1,+}$. If $X\in \left\lbrace X\in L^\infty\colon \exists\:\{X^i\}_{i\in\mathcal{I}}\in\mathbb{A}(X)\:s.t.\:X^i\in\partial\alpha_{\rho^i}^{min}(m)\:\forall\:i\in\mathcal{I}_\mu\right\rbrace$, let $\{X^i\}_{i\in\mathcal{I}}\in\mathbb{A}(X)$ such that  $X^i\in\partial\alpha_{\rho^i}^{min}(m)\:\forall\:i\in\mathcal{I}_\mu$. Then, \[E_m[-X]-\rho^\mu_{conv}(X)\geq\sum_{i\in\mathcal{I}}\left(E_m[-X^i]-\rho^i(X^i) \right)\mu_i=\sum_{i\in\mathcal{I}}\alpha_{\rho^i}^{min}(m)\mu_i\geq\alpha_{\rho^\mu_{conv}}^{min}(m). \] Thus, $X\in\partial\alpha_{\rho^\mu_{conv}}^{min}(m)$.
 		\end{enumerate}
 	\end{proof}
 
 \begin{Rmk}
 In the context of item (iii) and under coherence of $\rho_{\mathcal{I}}$,  by Proposition \ref{prp:bound} we have that $\rho(X)\leq\rho^\mu_{conv}(X)\leq\rho^i(X)\:\forall\:i\in\mathcal{I}_\mu,\:\forall\:X\in L^\infty$ for any convex risk measure $\rho$ such that $\rho(X)\leq \rho^i(X)\:\forall\:i\in\mathcal{I}_\mu,\:\forall\:X\in L^\infty$. In this sense, we can understand $\rho^\mu_{conv}$ as the ``lower-convexification" of the non-convex risk measure $\inf_{i\in\mathcal{I}_\mu}\rho^i$ in the sense that the former is the largest convex risk measure that is dominated by the latter. 
 	\end{Rmk}

 	We now present the main results regarding the representation of $\rho^\mu_{conv}$ for coherent cases. 
 	
 	\begin{Thm}\label{Thm:dualconv2}
 		Let $\rho_\mathcal{I}$ be a collection of coherent risk measures. Then,
 		\begin{enumerate}
 			
 			\item $\rho^{\mu}_{conv}$  is finite, and its dual set is  \begin{equation}
 			\mathcal{Q}_{\rho^{\mu}_{conv}} =\left\lbrace m\in ba_{1,+}\colon m\in\mathcal{Q}_{\rho^i}\:\forall\:i\in\mathcal{I}_\mu\right\rbrace=\bigcap_{i\in\mathcal{I}_\mu}\mathcal{Q}_{\rho^i}.
 			\end{equation}
 			In particular, $\mathcal{Q}_{\rho^{\mu}_{conv}}$ is non empty. 
 			\item  The acceptance set of $\rho^{\mu}_{conv}$ is  \begin{equation}
 			\mathcal{A}_{\rho^{\mu}_{conv}}=clconv\left(\mathcal{A}_\cup\right)=cl(\mathcal{A}_\mu),\:\mathcal{A}_\cup=\bigcup_{i\in\mathcal{I}_\mu}\mathcal{A}_{\rho^i}, 
 			\end{equation} where $clconv$ denotes the closed convex hull.
 		\end{enumerate}

 	\end{Thm}
 	
 	\begin{proof}
 		\begin{enumerate}
 			
 			\item  By Proposition \ref{prp:propconv}, we have that $\rho^\mu_{conv}$ is a  coherent risk measure. Since $\rho^\mu_{conv}(0)=0$, it is finite. In this case, its dual set is composed by the measures $m\in ba_{1,+}$ such that $\alpha_{\rho^\mu_{conv}}^{min}(m)=0$. By Theorems \ref{the:dual} and \ref{Thm:dualconv}, we obtain that \begin{align*}
 			\mathcal{Q}_{\rho^\mu_{conv}}&=\left\lbrace  m\in ba_{1,+}\colon\sum_{i\in\mathcal{I}}\alpha_{\rho^i}^{min}(m)\mu_i=0 \right\rbrace\\
 			&=\left\lbrace  m\in ba_{1,+}\colon \alpha_{\rho^i}^{min}(m)=0\:\forall\:i\in\mathcal{I}_\mu\right\rbrace \\
 			&=\left\lbrace  m\in ba_{1,+}\colon m\in\mathcal{Q}_{\rho^i}\:\forall\:i\in\mathcal{I}_\mu\right\rbrace=\bigcap_{i\in\mathcal{I}_\mu}\mathcal{Q}_{\rho^i}. 
 			\end{align*}
 			The convexity and closedness of $\mathcal{Q}_{\rho^\mu_{conv}}$ follow from the convexity and lower semicontinuity of $\alpha_{\rho^\mu_{conv}}^{min}$. Furthermore, if $\mathcal{Q}_{\rho^{\mu}_{conv}} = \left\lbrace  m\in ba_{1,+}\colon \alpha_{\rho^i}^{min}(m)=0\:\forall\:i\in\mathcal{I}_\mu\right\rbrace=\emptyset$, then by Corollary \ref{crl:hull}, $\rho^\mu_{conv}(0)<0$, which contradicts  coherence. 
 			\item  We recall that, by Theorem \ref{the:dual} and Remark \ref{rmk:dual}, for any  coherent risk measure $\rho\colon L^\infty\rightarrow\mathbb{R}$, we have $X\in\mathcal{A}_\rho$ if and only if $E_m[-X]\leq 0\:\forall\:m\in\mathcal{Q}_\rho$. Thus,
 			\begin{align*}
 			\mathcal{Q}_{\rho^\mu_{conv}}&=\bigcap_{i\in\mathcal{I}_\mu}\left\lbrace m\in ba_{1,+}\colon E_m[-X]\leq 0\:\forall\:X\in\mathcal{A}_{\rho^i}\right\rbrace\\
 			&=\left\lbrace m\in\mathcal{Q}\colon E_m[-X]\leq 0,\:\forall\:X\in\cup_{i\in\mathcal{I}_\mu}\mathcal{A}_{\rho^i}\right\rbrace\\
 			&=\left\lbrace m\in\mathcal{Q}\colon E_m[-X]\leq 0,\:\forall\:X\in clconv(\mathcal{A}_\cup)\right\rbrace=\mathcal{Q}_{\rho_{clconv(\mathcal{A}_\cup)}}.
 			\end{align*} 
 			By considering the closed convex hull does not affect is because the map $X\rightarrow E_m[X]$ is linear and continuous for any $m\in\mathcal{Q}$. Thus, $\rho^\mu_{conv}=\rho_{clconv(\mathcal{A}_\cup)}$. Hence, $\mathcal{A}_{\rho^{\mu}_{conv}}=\mathcal{A}_{\rho_{clconv\left(\mathcal{A}_\cup\right)}}=clconv\left(\mathcal{A}_\cup\right)$. We note that $\mathcal{A}_\cup$ is nonempty, monotone (in the sense that $X\in\mathcal{A}_\cup$ and $Y\geq X$ implies  $Y\in\mathcal{A}_\cup$), and a cone, as this is true for any $\mathcal{A}_{\rho^i}$. Moreover, it is evident that $\mathcal{A}_\cup\subseteq\mathcal{A}_\mu$ for normalized risk measures in $\rho_{\mathcal{I}}$. Thus, by Theorem \ref{Thm:dualconv}, we have that  $\mathcal{A}_{\rho^{\mu}_{conv}}=clconv\left(\mathcal{A}_\cup\right)\subseteq cl(\mathcal{A}_\mu)\subseteq\mathcal{A}_{\rho^{\mu}_{conv}}$. 
 		\end{enumerate}
 	\end{proof}

 	\begin{Rmk}\label{rmk:fatou2}
 	\begin{enumerate}
 		\item Similarly to Remark \ref{rmk:fatou}, under the assumption of Fatou continuity for both the risk measures in $\rho_{\mathcal{I}}$ and $\rho^\mu_{conv}$, the claims in Theorem \ref{Thm:dualconv} could be adapted by replacing the finitely additive measures $m\in ba_{+,1}$ by probabilities $\mathbb{Q}\in\mathcal{Q}$. Moreover, weak$^*$ topological concepts could replace the corresponding strong (norm) topological concepts. 
 		\item In light of Corollary \ref{crl:hull}, we have, under the hypotheses of Theorem \ref{Thm:dualconv2}, that $\rho^{\mu}_{conv}$ is finite and normalized if and only if the condition $\left\lbrace  m\in ba_{1,+}\colon \alpha_{\rho^i}^{min}(m)=0\:\forall\:i\in\mathcal{I}_\mu\right\rbrace\not=\emptyset$ is satisfied. Moreover, both assertions are  equivalent to $\mathcal{A}_\mu$ not being dense in $L^\infty$. The intuition for $\mathcal{A}_\cup$ is that some position is acceptable
 		if it is acceptable for any relevant (in the $\mu$ sense) member of $\rho_\mathcal{I}$. 
 	\end{enumerate}
 	\end{Rmk}
 	
 		If law invariance is satisfied, as is the case in most practical applications, interesting features are present. In this paper, when dealing with law invariance we always assume that our base probability space $(\Omega,\mathcal{F},\mathbb{P})$ is atomless. We then introduce the following risk measures:
 			
 		\begin{Exm}\label{Exm:meas}
 		 			\begin{enumerate}
 				\item Value at risk (VaR): This is a Fatou-continuous, law-invariant, comonotone, monetary risk measure  defined as $VaR^\alpha(X)=-F_{X}^{-1}(\alpha),\:\alpha\in[0,1]$. We have the acceptance set  $\mathcal{A}_{VaR^\alpha}=\left\lbrace X\in L^\infty:\mathbb{P}(X<0)\leq\alpha\right\rbrace $.
 				\item Expected shortfall (ES): This is a Fatou-continuous, law-invariant, comonotone, coherent risk measure  defined as $ES^{\alpha}(X)=\frac{1}{\alpha}\int_0^\alpha VaR^s(X)ds,\:\alpha\in(0,1]$ and $ES^0(X)=VaR^0(X)=-\operatorname{ess}\inf X$. We have $\mathcal{A}_{ES^\alpha}=\left\lbrace X\in L^\infty:\int_0^\alpha VaR^s(X)ds\leq0\right\rbrace $ and  $\mathcal{Q}_{ES^\alpha}=\left\lbrace \mathbb{Q}\in\mathcal{Q} : \frac{d\mathbb{Q}}{d\mathbb{P}}\leq\frac{1}{\alpha} \right\rbrace$. Note that $ES^1(X)=-E[X]$.
 			\end{enumerate}
 		\end{Exm}

 	
 	
 	\begin{Thm}[Theorems 4 and 7 in \cite{Kusuoka2001}, Theorem 4.1 in \cite{Acerbi2002a}, Theorem 7 in \cite{Fritelli2005}]\label{thm:dual2}
 		Let $\rho : L^\infty\rightarrow \mathbb{R}$ be a risk measure. Then
 		\begin{enumerate}
 			\item $\rho$ is a law-invariant, convex risk measure if and only if it can be represented as
 			\begin{equation}\label{eq:dual2}
 			\rho(X)=\sup\limits_{m\in\mathcal{M}}\left\lbrace \int_{(0,1]}ES^\alpha(X)dm-\beta^{min}_{\rho}(m)\right\rbrace,\:\forall\:X\in L^\infty,
 			\end{equation}
 			where $\mathcal{M}$ is the set of probability measures on $(0,1]$, and $\beta^{min}_{\rho} : \mathcal{M}\rightarrow\mathbb{R}_+\cup\{\infty\}$ is defined as
 			$\beta^{min}_{\rho}(m)=\sup\limits_{X\in\mathcal{A}_{\rho}} \int_{(0,1]}ES^\alpha(X)dm$. 
 			\item $\rho$ is a law-invariant, coherent risk measure if and only if it can be represented as
 			\begin{equation}\label{eq:dual2coh}
 			\rho(X)=\sup\limits_{m\in\mathcal{M}_\rho} \int_{(0,1]}ES^\alpha(X)dm,\:\forall\:X\in L^\infty,
 			\end{equation}
 			where  $\mathcal{M}_{\rho}=\left\lbrace m\in\mathcal{M}\colon\int_{(u,1]}\frac{1}{v}dm=F^{-1}_{\frac{d\mathbb{Q}}{d\mathbb{P}}}(1-u),\:\mathbb{Q}\in\mathcal{Q}_\rho\right\rbrace $.
 			\item $\rho$ is a law-invariant, comonotone, coherent risk measure if and only if it can be represented as
 			\begin{align}
 			\rho(X)&= \int_{(0,1]}ES^\alpha(X)dm\label{eq:dual2com}\\
 			&=\int_0^1VaR^\alpha(X)\phi(\alpha)d\alpha\label{eq:dual2com2}\\
 			&=\int_{-\infty}^{0}(g(\mathbb{P}(-X\geq x)-1)dx+\int_{0}^{\infty}g(\mathbb{P}(-X\geq x))dx,\:\forall\:X\in L^\infty,\label{eq:dual2com3}
 			\end{align}
 			where $m\in\mathcal{M}_\rho$, $\phi:[0,1]\rightarrow\mathbb{R}_+$ is decreasing and right-continuous, with $\phi(1)=0$ and $\int_{0}^{1}\phi(u)du=1$, and $g:[0,1]\rightarrow[0,1]$, called distortion, is increasing and concave, with $g(0)=0$ and $g(1)=1$. We have that $\int_{(u,1]}\frac{1}{v}dm=\phi(u)=g^\prime_+(u)\:\forall\:u\in[0,1]$.
 		\end{enumerate}
 	\end{Thm}
 	
 	\begin{Rmk}\label{rmk:LIdual}
 		\begin{enumerate}
 			\item 	We have that law-invariant convex risk measures $\rho$ are Fatou continuous, see Theorem 2.1 in \cite{Jouini2006} and Proposition 1.1 in \cite{Svindland2010}, and  $X\succeq Y$ implies that $\rho(X)\leq\rho(Y)$, see 	Theorem 4.3 in \cite{Bauerle2006} and Corollary 4.65 in \cite{Follmer2016}.
 			\item Functionals with representation as in (iii) of the last theorem are called spectral or distortion risk measures. This concept is related to capacity set functions and Choquet integrals. Note that Comonotonic Additivity implies coherence for convex risk measures, see Lemma 4.83 of \cite{Follmer2016} for instance. In this case, $\rho$ can be represented by $\mathcal{Q}_\rho=\{\mathbb{Q}\in\mathcal{Q}\colon\mathbb{Q}(A)\leq g(\mathbb{P}(A)),\:\forall\:A\in\mathcal{F}\}$, which is the core of $g$, if and only if $g$ is its distortion function. If $\phi$ is not decreasing (and thus $g$ is not concave), then the risk measure is not convex and cannot be represented as combinations of ES.
 			\item Without law invariance, we can (see, for instance, Theorem 4.94 and Corollary 4.95 in \cite{Follmer2016}) represent a convex, comonotone risk measure $\rho$ by a Choquet integral as follows: \begin{align*}
 			\rho^\mu_{conv}(X)&=\int(-X)dc\\
 			&=\int_{-\infty}^{0}(c(-X\geq x)-1)dx+\int_{0}^{\infty}c(-X\geq x)dx\\
 			&=\max\limits_{m\in ba_{1,+}^c}E_m[-X],\:\forall\:X\in L^\infty,
 			\end{align*} where $c\colon\mathcal{F}\rightarrow[0,1]$ is a normalized ($c(\emptyset)=0$ and $c(\Omega)=1$), monotone (if $A\subseteq B$ then $c(A)\leq c(B)$), submodular ($c(A\cup B)+c(A\cap B)\leq c(A)+c(B)$) set function that is called capacity and is defined as $c(A)=\rho(-1_A)\:\forall\:A\in\mathcal{F}$, and $ba_{1,+}^c=\{m\in ba_{1,+}\colon m(A)\leq c(A)\:\forall\:A\in\mathcal{F}\}$.
 		\end{enumerate}
 	\end{Rmk}

 	We now focus on dual representations under the assumption of law invariance and comonotonic additivity.

 	\begin{Thm}\label{Thm:dualLI}
 		Let $\rho_\mathcal{I}$ be a collection of convex, law-invariant risk measures. Then,
 		\begin{enumerate}
 			\item  $\rho^\mu_{conv}$ is finite, normalized and has penalty term \begin{equation}\label{eq:penalty2}
 			\beta^{min}_{\rho^{\mu}_{conv}}(m)=\sum_{i\in\mathcal{I}}\beta^{min}_{\rho^i}(m)\mu_i,\:\forall\:m\in\mathcal{M}.
 			\end{equation} 
 			\item 	If, in addition, $\rho_\mathcal{I}$ consists of comonotone risk measures, then $\rho^{\mu}_{conv}$  is comonotone, and its distortion function is \begin{equation}\label{eq:penalty3}
 			g=\inf_{i\in\mathcal{I}_\mu} g^i,
 			\end{equation} 
 			where $g^i$ is the distortion of $\rho^i$ for each $i\in\mathcal{I}$.
 		\end{enumerate}
 		
 	\end{Thm}
 	
 	\begin{proof}
 		\begin{enumerate}
 			\item By Theorem \ref{thm:dual2} and Proposition \ref{prp:propconv}, we have that $\rho^\mu_{conv}$ is a Fatou continuous convex risk measure with $\rho^\mu_{conv}<\infty$. Moreover, $\rho^\mu_{conv}$ is either finite or identically $-\infty$. Regarding finiteness and normalization, we note that $\rho(X)\geq -E[X],\:\forall\:X\in L^\infty$ for normalized, convex, law-invariant risk measures by second-order stochastic dominance. Thus, \[\alpha^{min}_{\rho^i}(\mathbb{P})=\sup\limits_{X\in L^\infty}\{E[-X]-\rho(X)\}\leq\sup\limits_{X\in L^\infty}\{\rho(X)-\rho(X)\}=0,\:\forall\:i\in\mathcal{I}.\] By the non-negativity of the penalty terms, we have $\alpha^{min}_{\rho_i}(\mathbb{P})=0,\:\forall\:i\in\mathcal{I}$. Hence, by item (ii) of Corollary \ref{crl:hull}, we conclude that $\rho^\mu_{conv}$ is  normalized and, consequently, finite. Moreover, the penalty term can be obtained by an argument similar to that in (ii) of Theorem \ref{Thm:dualconv} by considering the $m\rightarrow\int_{(0,1]}ES^\alpha(X)dm$ linear and playing the role of $m\rightarrow E_{m}[-X]$ and recalling that acceptance sets of law invariant risk measures are law invariant in the sense that $X\in\mathcal{A}_\rho$ and $X\sim Y$ implies $Y\in\mathcal{A}_\rho$.
 			\item By Theorems \ref{thm:dual2}, \ref{Thm:dualconv}, and \ref{Thm:dualconv2}, as well as Remark \ref{rmk:LIdual}, we have, recalling that $\rho^\mu_{conv}$ is finite and Fatou continuous, that \begin{align*}
 			\mathcal{Q}_{\rho^\mu_{conv}}&=\{\mathbb{Q}\in\mathcal{Q}\colon\mathbb{Q}(A)\leq g^i(\mathbb{P}(A))\:\forall\:i\in\mathcal{I}_\mu\:\forall\:A\in\mathcal{F}\}\\
 			&=\left\lbrace \mathbb{Q}\in\mathcal{Q}\colon\mathbb{Q}(A)\leq \inf_{i\in\mathcal{I}_\mu} g^i(\mathbb{P}(A))\:\forall\:A\in\mathcal{F}\right\rbrace.
 			\end{align*}
 			By the properties of the infimum and $\{g^i\}_{i\in\mathcal{I}}$, we obtain that $g:[0,1]\rightarrow[0,1]$ is increasing and concave, and it satisfies $g(0)=0$ and $g(1)=1$. Thus, $\rho^\mu_{conv}$ can be represented as a Choquet integral using \eqref{eq:dual2com3}, which implies that it is comonotone.
 		\end{enumerate}	
 	\end{proof}

 	Regarding comonotonic additivity, (ii) in Theorem \ref{Thm:dualLI} remains true if we drop the law invariance of $\rho_{\mathcal{I}}$, as shown in the following corollary.
 	
 	\begin{Crl}\label{crl:dualLI}
 		Let $\rho_\mathcal{I}$ be a collection of convex, comonotone risk measures. Then, $\rho^\mu_{conv}$ is finite, normalized, and comonotone, and its capacity function is \begin{equation}\label{eq:penalty4}
 		c(A)=\inf_{i\in\mathcal{I}_\mu} c^i(A),\:\forall\:A\in\mathcal{F},
 		\end{equation} 
 		where $c^i$ is the capacity of $\rho^i$ for each $i\in\mathcal{I}$.
 	\end{Crl}
 	
 	\begin{proof}
 		We note that, by Proposition \ref{prp:propconv}, we have $\rho^{\mu}_{conv}<\infty$ and $\rho^{\mu}_{conv}(0)=0>-\infty$. Let the set function $c\colon\mathcal{F}\rightarrow[0,1]$ be defined as $c(A)=\inf_{i\in\mathcal{I}_\mu}c^i(A)$, where $c^i$ is the capacity related to $\rho^i$ for each $i\in\mathcal{I}$. Thus, by an argument similar to that in Theorem \ref{Thm:dualLI}, if we consider $ba^c_{1,+}=\left\lbrace m\in ba_{1,+}\colon m(A)\leq c(A)\:\forall\:A\in\mathcal{F}\right\rbrace$, then the reasoning in Remark \ref{rmk:LIdual} implies that the claim is true. 
 	\end{proof}
 	
 	\section{Optimal allocations}\label{sec:alloc}

 	An exciting feature of traditional finite inf-convolution is capital allocation. A highly relevant concept is Pareto optimality, which is defined as follows.
 	
 	\begin{Def}
 		We call $\{X^i\}_{i\in\mathcal{I}}\in\mathbb{A}(X)$
 		\begin{enumerate}
 			\item  Optimal for $X\in L^\infty$ if $\sum_{i\in\mathcal{I}}\rho^i(X^i)\mu_i=\rho^\mu_{conv}(X)$.
 			\item  Pareto optimal for $X\in L^\infty$ if for any $\{Y^i\}_{i\in\mathcal{I}}\in\mathbb{A}(X)$  such that $\rho^i(Y^i)\leq\rho^i(X^i)\:\forall\:i\in\mathcal{I}_\mu$, we have $\rho^i(Y^i)=\rho^i(X^i)\:\forall\:i\in\mathcal{I}_\mu$.
 		\end{enumerate}
 		
 	\end{Def}
 	
 	\begin{Rmk}\label{rmk:opt}
 		\begin{enumerate}
 			\item If $\rho^\mu_{conv}$ is normalized, $\{X^i=0\}_{i\in\mathcal{I}}$ is optimal for $0$. This implies the condition that if $\{X^i\}_{i\in\mathcal{I}}\in\mathbb{A}(0)$ and $\rho^i(X^i)\leq0\:\forall\:i\in\mathcal{I}_\mu$, then $\rho^i(X^i)=0\:\forall\:i\in\mathcal{I}_\mu$, and therefore $\{X^i=0\}_{i\in\mathcal{I}}$ is also Pareto optimal for $0$. This can be understood as a non-arbitrage condition.
 			We note that any optimal allocation must be Pareto optimal and that a risk sharing rule is also a Pareto-optimal allocation. 
 		\end{enumerate}
 	\end{Rmk}

 	If $\rho_{\mathcal{I}}$ consists of monetary risk measures and $\mathcal{I}$ is finite, Theorem 3.1 in \cite{Jouini2008} shows that optimal and Pareto-optimal allocations coincide. In the following proposition, we extend this result to the context of countable $\mathcal{I}$.  
 	
 	\begin{Prp}\label{Prp:pareto}
 		 We have that
 		\begin{enumerate}
 			\item $\{X^i\}_{i\in\mathcal{I}}\in\mathbb{A}(X)$ is optimal for $X\in L^\infty$ if and only if it is Pareto optimal for $X\in L^\infty$.
 			\item if $\{X^i\}_{i\in\mathcal{I}}$ is optimal for $X\in L^\infty$, then so is $\{X^i+C^i\}_{i\in\mathcal{I}}$, where $C^i\in\mathbb{R}\:\forall\:i\in\mathcal{I}_\mu$, and $\sum_{i\in\mathcal{I}}C^i\mu_i=0$. 
 			\item Under sub-additivity of $\rho_{\mathcal{I}}$  we have that if $\{X^i\}_{i\in\mathcal{I}}$ is optimal for $X\in L^\infty$, then so is $\{X^i+Y^i\}$ for any $\{Y^i\}$ that is optimal for $0$.
 		 		\end{enumerate}  
 	\end{Prp}
 	
 	\begin{proof}
 		We prove (i) since both (ii) and (iii) are directly obtained.	The ``only if'' part is straightforward, as in Remark \ref{rmk:opt}.	For the ``if'' part, let $\{X^i\}_{i\in\mathcal{I}}\in\mathbb{A}(X)$ be not optimal for $X\in L^\infty$. Then, there is $\{Y^i\}_{i\in\mathcal{I}}\in\mathbb{A}(X)$ such that $\sum_{i\in\mathcal{I}}\rho^i(Y^i)\mu_i<\sum_{i\in\mathcal{I}}\rho^i(X^i)\mu_i$. Let $k^i=\rho^i(X^i)-\rho^i(Y^i)\:\forall\:i\in\mathcal{I}_\mu$ and $k=\sum_{i\in\mathcal{I}}k^i\mu_i>0$. Moreover, $\{Y^i-k^i+k\}_{i\in\mathcal{I}}\in\mathbb{A}(X)$ and \[\sum_{i\in\mathcal{I}}\rho^i(Y^i-k^i+k)\mu_i<\sum_{i\in\mathcal{I}}\rho^i(Y^i-\rho^i(X^i)+\rho^i(Y^i))\mu_i=\sum_{i\in\mathcal{I}}\rho^i(X^i)\mu_i.\]
 			Hence, $\{X^i\}_{i\in\mathcal{I}}$ is not Pareto optimal for $X$.
 	 	\end{proof}

 	We now determine a necessary and sufficient condition for optimality in the case of convex risk measures.
 	
 	\begin{Thm}\label{Thm:pareto}
 		Let $\rho_\mathcal{I}$ be a family of convex risk measures. Then $\{X^i\}_{i\in\mathcal{I}}\in\mathbb{A}(X)$ is optimal for $X\in L^\infty$ if and only if $\bigcap_{i\in\mathcal{I}_\mu}\partial\rho^i(X^i)=\partial\rho^\mu_{conv}(X)\not=\emptyset$.
 	 
 	\end{Thm}
 	
 	\begin{proof}
 	By Corollary \ref{crl:hull} we have $\bigcap_{i\in\mathcal{I}_\mu}\partial\rho^i(X^i)\subseteq\partial\rho^\mu_{conv}(X)$.	We assume $\bigcap_{i\in\mathcal{I}_\mu}\partial\rho^i(X^i)\not=\emptyset$. Then, let $m^\prime\in\left\lbrace m\in ba_{1,+}\colon m\in\partial\rho^i(X^i)\:\forall\:i\in\mathcal{I}_\mu \right\rbrace\subseteq\partial\rho^\mu_{conv}(X)$. We have that \[\sum_{i\in \mathcal{I}}\rho^i(X^i)\mu_i=\sum_{i\in\mathcal{I}}\left( E_{m^\prime}[-X^i]-\alpha_{\rho^i}^{min}(m^\prime)\right) \mu_i\leq E_{m^\prime}[-X]-\alpha_{\rho^\mu_{conv}}^{min}(m^\prime)=\rho^\mu_{conv}(X).\] Hence, $\{X^i\}_{i\in\mathcal{I}}$ is optimal for $X\in L^\infty$.	Regarding the converse, for any $m^\prime\in\partial\rho^\mu_{conv}(X)$, we obtain by an argument similar to that in Theorem \ref{Thm:dualconv} the following: \begin{align*}
 		\rho^\mu_{conv}(X)&=E_{m^\prime}[-X]-\sum_{i\in\mathcal{I}}\alpha^{min}_{\rho^i}(m^\prime)\mu_i\\
 		&=\lim\limits_{n\rightarrow\infty}\left( E_{m^\prime}[-X]+\sum_{i=1}^n\inf\limits_{Y\in L^\infty}\left\lbrace  E_{m^\prime}[Y]+\rho^i\left(Y\right)\right\rbrace \mu_i\right)  \\
 	&\leq\lim\limits_{n\rightarrow\infty}\left( E_{m^\prime}[-X]+\inf\limits_{\left\lbrace Y_1,\dots,Y_n\right\rbrace \subset L^\infty}\sum_{i=1}^n\left\lbrace  E_{m^\prime}[Y^i]+\rho^i\left(Y^i\right)\right\rbrace \mu_i\right)  \\
 	&\leq\lim\limits_{n\rightarrow\infty}\inf\left\lbrace \sum_{i=1}^n\left\lbrace  E_{m^\prime}[Y^i-X]+\rho^i\left(Y^i\right)\right\rbrace \mu_i\colon\sum_{i=1}^nY^i\mu_i=X\right\rbrace \\
  	&=\lim\limits_{n\rightarrow\infty}\inf\left\lbrace \sum_{i=1}^n\left\lbrace  \rho^i\left(Y^i\right)\right\rbrace \mu_i\colon\sum_{i=1}^nY^i\mu_i=X\right\rbrace  =\rho^\mu_{conv}(X).
 		\end{align*}
 		Thus, $\rho^\mu_{conv}(X)=\sum_{i\in\mathcal{I}_\mu}\rho^i(X^i)\mu_i$ if and only if $\exists\:m^\prime\in ba_{1,+}$ such that $\rho^i(X^i)=E_{m^\prime}[-X]+\alpha^{min}_{\rho^i}(m^\prime)\:\forall\:i\in\mathcal{I}_\mu$. Hence, $m^\prime\in\left\lbrace m\in ba_{1,+}\colon m\in\partial\rho^i(X^i)\:\forall\:i\in\mathcal{I}_\mu\right\rbrace$.
 	\end{proof}
 	
 	We have the following corollary regarding subdifferential and optimality conditions.
 	
 	\begin{Crl}\label{crl:como}
 		Let  $\rho_\mathcal{I}$ be a collection of convex risk measures. If for any $X\in L^\infty$ there is an optimal allocation, then $\partial\alpha_{\rho^\mu_{conv}}^{min}(m)=\sum_{i\in\mathcal{I}_\mu}\partial\alpha_{\rho^i}^{min}(m)\mu_i$  for any $m\in ba_{1,+}$.
 	\end{Crl}
 	\begin{proof}
 		From Corollary \ref{crl:hull} we have  $\sum_{i\in\mathcal{I}_\mu}\partial\alpha_{\rho^i}^{min}(m)\mu_i\subseteq\partial\alpha_{\rho^\mu_{conv}}^{min}(m)$ for any $m\in ba_{1,+}$. For the converse relation, if $\partial\alpha_{\rho^\mu_{conv}}^{min}(m)=\emptyset$, then the claim is immediately obtained. Let then $X\in\partial\alpha_{\rho^\mu_{conv}}^{min}(m)$. By the definition of Legendre--Fenchel
 		convex-conjugate duality, the optimality condition is equivalent to the existence of $m\in ba_{1,+}$ such that $X^i\in\partial\alpha^{min}_{\rho^i}(m)\:\forall\:i\in\mathcal{I}_\mu$. By Theorem \ref{Thm:pareto}, we have that $X^i\in\partial\alpha_{\rho^i}^{min}(m)\:\forall\:i\in\mathcal{I}_\mu$. Then, $X\in\sum_{i\in\mathcal{I}_\mu}\partial\alpha_{\rho^i}^{min}(m)\mu_i$.
 	\end{proof}

 	Under the assumption of law invariance, it is well known that, for finite $\mathcal{I}$,  the minimization problem has a solution under co-monotonic allocations (see, for instance, Theorem 3.2 in \cite{Jouini2008}, Proposition 5 in \cite{Dana2003}, or Theorem 10.46 in \cite{Ruschendor2013}). For the extension to general $\mathcal{I}$, we should extend some definitions and results regarding comonotonicity. We note that if $\mathcal{I}$ is finite, these are equivalent to their traditional counterparts.

 	\begin{Def}\label{def:como}
 		$\{X^i\}_{i\in\mathcal{I}}$ is called $\mathcal{I}$-comonotone if $(X^i,X^j)$ is comonotone $\forall\:(i,j)\in\mathcal{I}_\mu\times\mathcal{I}_\mu$.
 	\end{Def}

 	\begin{Lmm}\label{lmm:como}
 	$\{X^i\}_{i\in\mathcal{I}}\in \mathbb{A}(X)$, $X\in L^\infty$, is $\mathcal{I}$-comonotone if and only if there exists a class of functions $\{h^i\colon\mathbb{R}\rightarrow\mathbb{R},\: i\in\mathcal{I}\}$ that are ($\forall\:i\in\mathcal{I}_\mu$) Lipschitz continuous and increasing, and satisfy $X^i=h^i\left(X\right)$ and $\sum_{i\in\mathcal{I}}h^i(x)\mu_i=x,\:\forall\:x\in\mathbb{R}$. In particular, if $\{X^i\}_{i\in\mathcal{I}}\in\mathbb{A}(X)$ is $\mathcal{I}$-comonotone, then $F^{-1}_X(\alpha)=\sum_{i\in\mathcal{I}}F^{-1}_{X^i}(\alpha)\mu_i,\:\forall\:\alpha\in[0,1]$.
 	\end{Lmm}
 	
 	\begin{proof}
 		For the ``if'' part, let $X^i,X^j\in\mathcal{I}_\mu$ such that $X^i=h^i(X)$ and $X^j=h^j(X)$ for $\{h^i\}_i\in\mathcal{I}$ satisfying the assumptions. Then $X^i,X^j$ are comonotone. For the ``only if'' part, let 	$\{X^i\}_{i\in\mathcal{I}}$ be $\mathcal{I}$-comonotone and $X(\Omega)=\{x\in\mathbb{R}\colon\exists\:\omega\in\Omega\;\text{s.t.}\;X(w)=x\}$. Then, for any fixed $\omega\in\Omega$, there is a family $\{x^i=X^i(\omega)\in\mathbb{R}\colon i\in\mathcal{I}\}$ such that $X(\omega)=x=\sum_{i\in\mathcal{I}}x^i\mu_i\in X(\Omega)$. Moreover, we define $h^i(x)=x^i\:\forall\:i\in\mathcal{I}_\mu$. If there are $\omega,\omega^\prime\in\Omega$ such that $\sum_{i\in\mathcal{I}}X^i(\omega)\mu_i=x=\sum_{i\in\mathcal{I}}X^i(\omega^\prime)\mu_i$, we then obtain $\sum_{i\in\mathcal{I}}\left( X^i(\omega)-X^i(\omega^\prime)\right) \mu_i=0$. Assuming $\mathcal{I}$-comonotonicity, we have that $X^i(\omega)=X^i(\omega^\prime)\:\forall\:i\in\mathcal{I}_\mu$. Consequently, the map $x\rightarrow\sum_{i\in\mathcal{I}}h^i(x)\mu_i=Id(x)$ is well defined. Regarding the increasing behavior of $h^i$, let $x,y\in X(\Omega)$ with $x\leq y$. Then, there are $\omega,\omega^\prime$ such that $\sum_{i\in\mathcal{I}}X^i(\omega)\mu_i=x\leq y=\sum_{i\in\mathcal{I}}X^i(\omega^\prime)\mu_i$, which implies $\sum_{i\in\mathcal{I}}\left( X^i(\omega)-X^i(\omega^\prime)\right) \mu_i\leq0$. Comonotonicity implies that this relation is equivalent to $h^i(x)=X^i(\omega)\leq X^i(\omega^\prime)=h^i(y)\:\forall\:i\in\mathcal{I}_\mu$. Concerning Lipschitz continuity, for any  $x,x+\delta\in X(\Omega)$ with $\delta>0$  we obtain $0\leq h^i(x+\delta)-h^i(x)\leq(\mu_i)^{-1}\delta,\:\forall\:i\in\mathcal{I}_\mu$. The first inequality is due to the increasing behavior of $h^i$. The second is because for any $i\in\mathcal{I}$ we have $x+\delta=\sum_{j\in \mathcal{I}\backslash\{i\}}h^j(x+\delta)\mu_j+h^i(x+\delta)\mu_i\geq\sum_{j\in \mathcal{I}\backslash\{i\}}h^j(x)\mu_j+h^i(x)\mu_i=h^i(x+\delta)\mu_i+x-h^i(x)\mu_i$. It remains to extend $\{h^i\}$ from $X(\Omega)$ to $\mathbb{R}$. We first extend it to $cl(X(\Omega))$. If $x\in bd(X(\Omega))$ is only a one-sided
 		boundary point, then the continuous extension poses no problem, as increasing functions are involved. If $x$ can be approximated from both sides, then Lipschitz continuity implies that the left- and right-sided continuous extensions coincide. The extension to $\mathbb{R}$ is performed linearly in each connected component of
 		$\mathbb{R}\backslash cl(X(\Omega))$ so that the condition $\sum_{i\in\mathcal{I}}h^i(x)=x$ is satisfied. Then, the main claim is proved. Moreover, let $\{X^i\}_{i\in\mathcal{I}}\in\mathbb{A}(X)$ be $\mathcal{I}$-comonotone. Then, $x\rightarrow\sum_{i\in\mathcal{I}}h^i(x)\mu_i$ is also Lipschitz continuous and increasing. We recall that $F^{-1}_{g(X)}=g(F^{-1}_X)$ for any increasing function $g\colon\mathbb{R}\rightarrow\mathbb{R}$. Then, for any $\alpha\in[0,1]$, we obtain 
 		\[F^{-1}_X(\alpha)=F^{-1}_{\sum_{i\in\mathcal{I}}h^i(X)\mu_i}(\alpha)=\sum_{i\in\mathcal{I}}h^i(F^{-1}_X(\alpha))\mu_i=\sum_{i\in\mathcal{I}}F^{-1}_{h^i(X)}(\alpha)\mu_i=\sum_{i\in\mathcal{I}}F^{-1}_{X^i}(\alpha)\mu_i.\] 
 	\end{proof}
 	
 	We now prove the following comonotonic-improvement Theorem for arbitrary $\mathcal{I}$.
 	
 	\begin{Thm}\label{Lmm:como2}
 		Let $X\in L^\infty$. Then, for any $\{X^i\}_{i\in\mathcal{I}}\in\mathbb{A}(X)$, there is an $\mathcal{I}$-comonotone $\{Y^i\}_{i\in\mathcal{I}}\in\mathbb{A}(X)$ such that $Y^i\succeq X^i\:\forall\:i\in\mathcal{I}_\mu$.
 	\end{Thm}
 	
 	\begin{proof}
 		Let $\mathcal{F}_n$ be the $\sigma$-algebra generated by $\{\omega\colon k2^{-n}\leq X(\omega)\leq k2^n\}\subset\Omega$ for $k>0$, $X_n=E[X|\mathcal{F}_n]$, and $X^i_n=E[X^i|\mathcal{F}_n]\:\forall\:i\in\mathcal{I}$. Then, $\lim\limits_{n\rightarrow\infty}X_n=X$, $\lim\limits_{n\rightarrow\infty}X_n^i=X^i$ for any $i\in\mathcal{I}$, and $X^i_n\succeq X^i$ for each $n$ each $i$. By the arguments in Proposition 1 in \cite{Landsberger1994} or Proposition 10.46 in \cite{Ruschendor2013}, we can conclude that every
 		allocation of $X$ taking a countable number of values is dominated by a comonotone allocation. Thus, by Lemma \ref{lmm:como}, for any $n\in\mathbb{N}$, there are Lipschitz continuous, increasing functions $\{h^i_n\colon\mathbb{R}\rightarrow\mathbb{R},\: i\in\mathcal{I}\}$ with $\sum_{i\in\mathcal{I}}h^i_n\mu_i=Id$ such that $Y^i_n=h^i_n(X_n)\succeq X^i_n\:\forall\:i\in\mathcal{I}_\mu$. We note that these functions constitute a bounded, closed, equicontinuous family. Then, by Ascoli's theorem, there is a subsequence of $\{h^i_n\}$ that converges uniformly on $[\operatorname{ess}\inf X,\operatorname{ess}\sup X]$ to the Lipschitz-continuous and increasing $h^i$ in the $\forall\:i\in\mathcal{I}_\mu$ sense. Thus, $\sum_{i\in\mathcal{I}}h^i\mu_i=Id$ on $[\operatorname{ess}\inf X,\operatorname{ess}\sup X]$. We then have $Y^i=h^i(X)\succeq X^i\:\forall\:i\in\mathcal{I}_\mu$ by considering uniform limits. Finally, by Lemma \ref{lmm:como}, we obtain that $\{Y^i\}_{\:i\in\mathcal{I}}$ is  $\mathcal{I}$-comonotone. It remains to show that $\{Y^i\}_{\:i\in\mathcal{I}}$ belongs to $\mathbb{A}(X)$. We have that $\sum_{i\in\mathcal{I}}Y^i\mu_i=\sum_{i\in\mathcal{I}}h^i(X)\mu_i=X\:\mathbb{P}-a.s.$  Hence, $\{Y^i\}_{\:i\in\mathcal{I}}\in\mathbb{A}(X)$.
 	\end{proof}

 	We are now in a position to extend the existence of optimal allocations to our framework of law-invariant, convex risk measures.
 	
 	\begin{Thm}\label{Thm:como}
 		Let  $\rho_\mathcal{I}$ be a collection of law-invariant, convex risk measures. Then,
 		\begin{enumerate}
 			\item For any $X\in L^\infty$, there is an $\mathcal{I}$-comonotone optimal allocation. 
 			\item In addition to initial hypotheses, if $\rho_{\mathcal{I}}$ consists of risk measures that are strictly monotone with respect to $\succeq$, then every optimal allocation for any $X\in L^\infty$  is $\mathcal{I}$-comonotone.
 			\item  In addition to initial hypotheses, if $\rho_{\mathcal{I}}$ consists of strictly convex functionals, we have uniqueness of the optimal allocation up to scaling (if $\{X^i\}_{i\in\mathcal{I}}$ is optimal for $X\in L^\infty$, then so is $\{X^i+C^i\}_{i\in\mathcal{I}}$, where $C^i\in\mathbb{R}\:\forall\:i\in\mathcal{I}_\mu$ and $\sum_{i\in\mathcal{I}}C^i\mu_i=0$).	
 		\end{enumerate}
 	\end{Thm}
 	
 	\begin{proof}
 		\begin{enumerate}
 			\item By Theorem \ref{Lmm:como2}, we can restrict the minimization problem to $\mathcal{I}$-comonotonic allocations, as, by Theorem \ref{thm:dual2}, law-invariant risk measures preserve second-order stochastic dominance. Let $\{Y^i_n=h^i_n(X)\in L^\infty,\:i\in\mathcal{I}\}_n$ be an optimal sequence for $X$, i.e. $\lim\limits_{n\rightarrow\infty}\sum_{i\in\mathcal{I}}\rho^i(Y^i_n)\mu_i=\rho^\mu_{conv}(X)$, where $h^i_n\colon[\operatorname{ess}\inf X,\operatorname{ess}\sup X]\rightarrow\mathbb{R}$ are increasing, bounded, and Lipschitz-continuous functions. Such sequence always exist because $\rho_\mathcal{I}$ is monetary and we can take $h^i_n=-x_n$ for any $i\in\mathcal{I}$, where $x_n\to\rho^\mu_{conv}(X)$. Thus, each $h^i_n$ is increasing, bounded and Lipschitz continuous while $\lim\limits_{n\rightarrow\infty}\sum_{i\in\mathcal{I}}\rho^i(h^i_n(X))\mu_i=\lim\limits_{n\rightarrow\infty}\sum_{i\in\mathcal{I}}\rho^i(-x_n)\mu_i=\lim\limits_{n\rightarrow\infty}x_n=x$. By an argument similar to that in Theorem \ref{Lmm:como2}, we have that $h^i$ is the uniform limit (after passing to a subsequence if necessary) of $\{h^i_n\}$. Thus, $Y^i_n=h^i_n(X)\rightarrow h^i(X)=Y^i$. By continuity in the essential supremum norm, we have that $\lim\limits_{n\rightarrow\infty}\left|\rho^i(Y^i_n)-\rho^i(Y^i)\right|=0$. As $h^i$ is the uniform limit of $\{h^i_n\}$, we have by dominated convergence, since $|\rho^i(Y^i_n)|\leq\lVert X\rVert_\infty<\infty$ for any $n\in\mathbb{N}$, that \[\rho^\mu_{conv}(X)=\lim\limits_{n\rightarrow\infty}\sum_{i\in\mathcal{I}}\rho^i(Y^i_n)\mu_i=\sum_{i\in\mathcal{I}}\lim\limits_{n\rightarrow\infty}\rho^i(Y^i_n)\mu_i=\sum_{i\in\mathcal{I}}\rho^i(Y^i)\mu_i.\]
 			Hence, $\{Y^i\}_{i\in\mathcal{I}}$ is the desired optimal allocation.
 			\item We recall that strict monotonicity implies that if $X\succeq Y$ and $X\not\sim Y$, then $\rho^i(X)<\rho^i(Y)\:\forall\:i\in\mathcal{I}_\mu$ for any $X,Y\in L^\infty$. Let $\{X^i\}_{i\in\mathcal{I}}$ be an optimal allocation for $X\in L^\infty$. Then, by Theorem \ref{Lmm:como2}, there is an $\mathcal{I}$-comonotone allocation $\{Y^i\}_{i\in\mathcal{I}}\in\mathbb{A}(X)$ such that \[\rho^\mu_{conv}(X)=\sum_{i\in\mathcal{I}}\rho^i(X^i)\mu_i\geq\sum_{i\in\mathcal{I}}\rho^i(Y^i)\mu_i.\] Thus, $\{Y^i\}_{i\in\mathcal{I}}$ is also optimal. If $X^i=Y^i\:\forall\:i\in\mathcal{I}_\mu$, then we have the claim. If there is $i\in\mathcal{I}_\mu$ such that  $X^i\not=Y^i$, then we have by strictly monotonicity regarding $\succeq$ that $\rho^i \left(Y^i\right)<\rho^i(X^i)$, contradicting the optimality of $\{X^i\}_{i\in\mathcal{I}}$. Hence, every optimal allocation for $X$ is $\mathcal{I}$-comonotone.
 			\item  We assume, toward a contradiction, that both $\{X^i\}_{i\in\mathcal{I}}$ and $\{Y^i\}_{i\in\mathcal{I}}$ are optimal allocations for $X\in L^\infty$ such that there is $i\in\mathcal{I}_\mu$ with $X^i\not= Y^i$ and there is no $C^i\in\mathbb{R}\:\forall\:i\in\mathcal{I}_\mu$ such that $\sum_{i\in\mathcal{I}}C^i\mu_i=0$ and $\{X^i+C^i\}_{i\in\mathcal{I}}$ (otherwise, item (ii) in Proposition \ref{Prp:pareto} assures optimality). We note that for any $\lambda\in[0,1]$, the family $\{Z^i=\lambda X^i+(1-\lambda) Y^i\}_{i\in\mathcal{I}}$ is in $\mathbb{A}(X)$. However, we would have \[\sum_{i\in\mathcal{I}}\rho^i(Z^i)\mu_i<\lambda\sum_{i\in\mathcal{I}}\rho^i(X^i)\mu_i+(1-\lambda)\sum_{i\in\mathcal{I}}\rho^i(Y^i)\mu_i=\rho^\mu_{conv}(X),\]
 			which contradicts the optimality of both $\{X^i\}_{i\in\mathcal{I}}$ and $\{Y^i\}_{i\in\mathcal{I}}$ for $X$.
 		\end{enumerate}
 	\end{proof}
 	
 	\begin{Rmk}
 		The examples in \cite{Jouini2008}, and \cite{Delbaen2006} show that law invariance is essential to ensure the existence of an $\mathcal{I}$-comonotone solution as above. However, the uniqueness of this optimal allocation is not ensured outside the scope of strict convexity  as in item (iii) of the last Theorem.
 	\end{Rmk}

 	If $\rho_{\mathcal{I}}$ consists of comonotone, law-invariant, convex risk measures, then we can prove an additional result regarding the connection between optimal allocations and the notion of flatness for quantile functions. To this end, the following definitions and Lemma are required. We recall that $dF^{-1}_X$ is the differential of $F^{-1}_X$.
 	
 	\begin{Def}\label{def:flat}
 		Let $g^1,g^2$ be two distortions with $g^1\leq g^2$. A quantile function $F^{-1}_X$, $X\in L^\infty$, is called flat on $\{x\in[0,1]\colon g^1(x)<g^2(x)\}$ if $dF^{-1}_X=0$ almost everywhere on  $\{g^1<g^2\}$ and $(F^{-1}_X(0^+)-F^{-1}_X(0))(g^2(0^+)-g^1(0^+))=0$.
 	\end{Def}

 	\begin{Lmm}[Lemmas 4.1 and 4.2 in \cite{Jouini2008}]\label{Lmm:jou}
 		Let $\rho\colon L^\infty\rightarrow\mathbb{R}$ be a law-invariant, comonotone, convex risk measure with distortion $g$. Moreover, let $m\in ba_{1,+}$ has a Lebesgue decomposition $m=Z_m\mathbb{P}+m^s$ into a regular part with density $Z_m$ and a singular part $m^s$. Then, \begin{enumerate}
 			\item $g_m\colon[0,1]\rightarrow\mathbb{R}$ defined as $g_m(0)=0$ and $g_m(t)=\lVert m^s\rVert_{TV}+\int_0^tF^{-1}_{Z_m}(1-s)ds,\:0<t\leq1$, is a concave distortion.
 			\item For any $m\in\partial\rho(X)$, we have that  $X$ and $-Z_m$ are comonotone. Moreover, the measure $m^\prime$ such that $Z_{m^\prime}=E[Z_m|X]$ belongs to $\partial\rho(X)$. 
 			\item $\partial\rho(X)=\left\lbrace m\in ba_{1,+}\colon g_m\leq g,\:F^{-1}_X\:\text{is flat on}\:\{g_m<g\}\right\rbrace $.
 		\end{enumerate}
 	\end{Lmm}

 	\begin{Thm}\label{Thm:como2}
 		Let  $\rho_\mathcal{I}$ consist of law-invariant, comonotone, convex risk measures with distortions $\{g^i\}_{i\in\mathcal{I}}$, $g=\inf_{i\in\mathcal{I}_\mu}g^i$, and $\{X^i\}_{i\in\mathcal{I}}\in\mathbb{A}(X)$ be $\mathcal{I}$-comonotone. If $F^{-1}_{X^i}$ is flat on $\{g<g^i\}\cap\{dF^{-1}_X>0\}\:\forall\:i\in\mathcal{I}_\mu$, then $\{X^i\}_{i\in\mathcal{I}}$ is an optimal allocation for $X\in L^\infty$. The converse is true if $(i,\alpha)\to VaR^\alpha(X^i)g^\prime_m(\alpha)$ is bounded.
 	\end{Thm}
 	
 	\begin{proof}
 		
 		By Theorem \ref{Thm:dualLI}, $g$ is the distortion of $\rho^\mu_{conv}$. Let $U$ be a $[0,1]$-uniform random variable (the existence of which is ensured because the space is atomless) such that $X=F^{-1}_X(U)$. We define $m\in ba_{1,+}$ by $m=g(0^+)\delta_0(U)+g^\prime(U)1_{(0,1]}(U)$, where $\delta_0$ is the Dirac measure at $0$. It is easily verified using (i) in Lemma \ref{Lmm:jou} that $g_m=g$. Moreover, let $\{X^i\}_{i\in\mathcal{I}}$ be an $\mathcal{I}$-comonotone optimal allocation for $X$ (the existence of which is ensured by Theorem \ref{Thm:como}). Thus, in the $\forall\:i\in\mathcal{I}_\mu$ sense, $g_m\leq g^i$, $-Z_m$ is comonotone with $X^i$, and, by hypothesis, $F^{-1}_{X^i}$ is flat on $\{g_m<g^i\}\cap\{dF^{-1}_X>0\}$. By Lemma \ref{lmm:como}, we have $\{dF^{-1}_X=0\}=\{\alpha\in[0,1]\colon dF^{-1}_{X^i}(\alpha)=0\;\forall\:i\in\mathcal{I}_\mu\}$. Thus, $\{g_m<g^i\}\cap\{dF^{-1}_X=0\}=\emptyset\:\forall\:i\in\mathcal{I}_\mu$. By (iii) of Lemma \ref{Lmm:jou}, we have that $m\in\partial\rho^i(X^i)\:\forall\:i\in\mathcal{I}_\mu$ Then, by Theorem \ref{Thm:pareto}, we obtain that $\{X^i\}_{i\in\mathcal{I}}$ is an optimal allocation. 
 		
 		For the converse, by Theorem \ref{Thm:pareto} and Corollary \ref{crl:como}, we have that $X^i\in\partial\alpha^{min}_{\rho^i}(m^\prime)\:\forall\:i\in\mathcal{I}_\mu$ for some $m^\prime\in ba_{1,+}$. By Corollary \ref{crl:como}, we have that $X\in\partial\alpha^{min}_{\rho^\mu_{conv}}(m^\prime)$, and by convex-conjugate duality, $m^\prime\in\partial\rho^\mu_{conv}(X)$. Thus, by (ii) in Lemma \ref{Lmm:jou}, we obtain that $m\in ba_{1,+}$ such that $Z_m=E[Z_{m^\prime}|X]$ belongs to $\partial\rho^\mu_{conv}(X)=\left\lbrace m\in ba_{1,+}\colon m\in\partial\rho^i(X^i)\:\forall\:i\in\mathcal{I}_\mu\right\rbrace$. By Theorem \ref{thm:dual2} and (iii) of Lemma \ref{Lmm:jou}, we have that $\rho^i(X^i)=\int_{0}^1VaR^\alpha(X^i)g_m^\prime(\alpha)d\alpha\:\forall\:i\in\mathcal{I}_\mu$. As $\{X^i\}_{i\in\mathcal{I}}$ is an optimal allocation, Theorem \ref{Thm:dualLI} and Lemma \ref{lmm:como} imply that
 		\begin{align*}
 		\int_0^1VaR^\alpha(X)g^\prime(\alpha)d\alpha&=\rho^\mu_{conv}(X)\\
 		&=\sum_{i\in\mathcal{I}}\int_{0}^1VaR^\alpha(X^i)g_m^\prime(\alpha)d\alpha \mu_i\\
 		&=\int_{0}^1\sum_{i\in\mathcal{I}}VaR^\alpha(X^i)\mu_i g_m^\prime(\alpha)d\alpha\\
 		&=\int_0^1VaR^\alpha(X)g^\prime_m(\alpha)d\alpha.
 		\end{align*}
 		We can make the interchange of sum and integral for dominated convergence because of the boudedness assumption. By continuity,  $VaR^\alpha(X)d\alpha = -dF^{-1}_X(\alpha)$. Then, we have $\int_0^1(g_m(\alpha)-g(\alpha))dF^{-1}_X(\alpha)=0$. As $m\in\partial\rho^\mu_{conv}(X)$, (iii) in Lemma \ref{Lmm:jou} implies that $g_m\leq g$, and therefore $g_m=g$ in $\{dF^{-1}_X>0\}$. Hence, $F^{-1}_{X^i}$ is flat on $\{g<g^i\}\cap\{dF^{-1}_X>0\}\:\forall\:i\in\mathcal{I}_\mu$.
 	\end{proof}
 	
 	\begin{Rmk}We also have in this context that for any optimal allocation $\{X^i\}_{i\in\mathcal{I}}$, $F^{-1}_{X^i}$ is flat on $\{g^i\not=g^j\}\cap\{dF^{-1}_X=0\}$ for any $j\not=i$ in $\mathcal{I}_\mu$ sense. To see this, let $t\in\{g^j<g^i\}\cap\{dF^{-1}_X=0\}$. We note that, by Lemma \ref{lmm:como}, $\{dF^{-1}_X=0\}=\{\alpha\in[0,1]\colon dF^{-1}_{X^i}(\alpha)=0\:\forall\:i\in\mathcal{I}_\mu\}$.  Then, by Lemma \ref{Lmm:jou}, $g_m(t)<g_i(t)$, and thus $dF^{-1}_{X^i}$ is flat at $\{g_m<g_i\}$. By comonotonicity and Lemma \ref{lmm:como}, the same is true for $dF^{-1}_{X^j}$. By repeating the argument for $t\in\{g^j>g^i\}\cap\{dF^{-1}_X=0\}$, we prove the claim. \end{Rmk}

 	We now follow the approach of \cite{Embrechts2018} by focusing on the robustness of optimal allocations instead of the one for $\rho^\mu_{conv}$. Intuitively, if an optimal allocation is robust, then the true aggregate risk value would be close to the obtained one under a small model misspecification. We now formally define such a concept and prove a result that relates robustness and upper semi continuity. To that, we need the concept of allocation principle, which is defined in the following. Note that from Lemma \ref{lmm:como}, $\mathcal{I}$-comonotonicity implies the existence of allocation principles. 
 	
 	\begin{Def}\label{def:rob}We define the following:\begin{enumerate}	\item $\{h^i\colon\mathbb{R}\rightarrow\mathbb{R},\:i\in\mathcal{I}\}$ is an allocation principle if  $\forall\:i\in\mathcal{I}_\mu$: $h^i(X)\in L^\infty\;\forall\:X\in L^\infty$, $h^i$ has at most finitely points of discontinuity, and $\sum_{i\in\mathcal{I}}h^i\mu_i$ is the identity function. We denote by $\mathbb{H}$ the set of allocation principles.
 	\item let $d$ be a pseudo-metric on $L^\infty$ and  $\{h^i\}_{i\in\mathcal{I}}\in\mathbb{H}$. Then $\{h^i(X)\}_{i\in\mathcal{I}}$ is $d$-robust if the map on $L^\infty$ defined as $Y\rightarrow\sum_{i\in\mathcal{I}}\rho^i(h^i(Y))\mu_i$ is continuous at $X$ in respect to $d$. \end{enumerate}\end{Def}

 	\begin{Prp}\label{prp:rob}	Let $d$ be a pseudo-metric on $L^\infty$. If $\{h^i(X)\}_{i\in\mathcal{I}}$ is a $d$-robust optimal allocation of $X\in L^\infty$, then $\rho^\mu_{conv}$ is upper semi continuous at $X$ with respect to $d$.\end{Prp}
 	
 	\begin{proof}	Let $\{h^i(X)\}_{i\in\mathcal{I}}$ be a $d$-robust optimal allocation for $X\in L^\infty$, and $\{X_n\}\subset L^\infty$ be such that $X_n\rightarrow X$ in $d$. Since $\rho^\mu_{conv}(X_n)\leq\sum_{i\in\mathcal{I}}\rho^i(h^i(X_n))\mu_i$, by convergence regarding $d$ we get that \[\limsup\limits_{n\rightarrow\infty}\rho^\mu_{conv}(X_n)\leq\limsup\limits_{n\rightarrow\infty}\sum_{i\in\mathcal{I}}\rho^i(h^i(X_n))\mu_i=\rho^\mu_{conv}(X).\]	\end{proof}
 	
 	\begin{Rmk}	Regarding the converse statement for proposition \ref{prp:rob}, Theorem 5 of \cite{Embrechts2018}, which is focused in quantiles, asserts that even continuity in respect to $d$ is not sufficient for the existence of a robust optimal allocation.  It is well known that convex risk measures are not upper semicontinuous with respect to the Levy metric. By  Proposition \ref{prp:propconv}, if each member of $\rho_\mathcal{I}$ is a convex risk measure, then so is $\rho^\mu_{conv}$. Thus, in this case, there are no robust optimal allocations regarding the Levy metric. In fact, this is a relatively new concept, and even in the finite $\mathcal{I}$ case, there is still not a general sufficient condition. We left for future research a more complete study of general sufficient conditions for robust allocations. \end{Rmk}

 	\section{Self-convolution and regulatory arbitrage}\label{sec:special}
 	
 	In this section, we consider the special case  $\rho^i=\rho,\:\forall\:i\in\mathcal{I}$. In this situation, we have that $\rho^\mu_{conv}$ is a self-convolution. This concept is fundamental in the context of regulatory arbitrage (as in \cite{Wang2016}), where the goal is to reduce the regulatory capital of a position by splitting it. More specifically, it indicates the gain,
 	in the sense of a decrease in capital requirements that a company achieves by splitting its
 	group structure into sub-companies within a given regulatory framework. The difference between $\rho(X)$ and $\rho^\mu_{conv}(X)$ is then obtained by a simple rearrangement (sharing) of risk. Thus, there is a reduction in capital requirements and consequently a decrease of opportunity and capital costs simply by arbitrating a splitting arrangement instead of changing a company's business structure. We now adjust this concept to our framework.
 	
 	\begin{Def}\label{def:arb}
 		The regulatory arbitrage of a risk measure $\rho$ is a functional $\tau_\rho\colon L^\infty\rightarrow\mathbb{R}_+\cup\{\infty\}$ defined as \begin{equation}\label{eq:arb}
 		\tau_\rho(X)=\rho(X)-\rho^\mu_{conv}(X),\:\forall\:X\in L^\infty.
 		\end{equation}
 		Moreover, $\rho$ is called
 		\begin{enumerate}
 			\item free of regulatory arbitrage if $\tau_\rho(X)=0,\:\forall\:X\in L^\infty$;
 			\item of finite regulatory arbitrage if $\tau_\rho(X)<\infty,\:\forall\:X\in L^\infty$; 
 			\item of partially infinite regulatory arbitrage if  $\tau_\rho(X)=\infty$ for some $X\in L^\infty$;
 			\item of infinite regulatory arbitrage if  $\tau_\rho(X)=\infty,\:\forall\:X\in L^\infty$. 
 		\end{enumerate}
 	\end{Def}
 	
 	\begin{Rmk}\label{Rmk:self}
 		\begin{enumerate}
 		 	\item	As $\rho^\mu_{conv}\leq\rho<\infty$, we have that $\tau_\rho$ is well defined. Our approach is different from that in \cite{Wang2016} because we consider ``convex" inf-convolutions instead of direct sums. More specifically, the approach by \cite{Wang2016} is  \begin{align*}
 		R(X)&=\inf\left\lbrace \sum_{i=1}^n\rho(X^i),n\in\mathbb{N},X^i\in L^\infty,i=1,\cdots,n,\sum_{i=1}^nX^i=X\right\rbrace\\
 		&=\lim\limits_{n\rightarrow\infty}\inf\left\lbrace \sum_{i=1}^n\rho(X^i),X^i\in L^\infty,i=1,\cdots,n,\sum_{i=1}^nX^i=X\right\rbrace, 
 		\end{align*}
 	while in our case we have the role for $\{\mu_ i\}$. This distinction leads to differences because in our approach convexity rules out regulatory arbitrage, while in \cite{Wang2016} sub-additivity plays this role. In fact, in the approach of \cite{Wang2016}, $\tau_\rho$ is always sub-additive, whereas in our case, this is not ensured. This fact alters most  results and arguments, as it is crucial in his study. For instance, in that approach, a convex but not coherent risk measure $\rho$ is of limited regulatory arbitrage, whereas in ours (see Theorem \ref{prp:arb} below) it is free of regulatory arbitrage. 
 	\item Most results below would remain true if we considered the general framework of arbitrary $\rho_{\mathcal{I}}$ and made the adaption $\tau_{\rho_\mathcal{I}}=\sum_{i\in\mathcal{I}}\rho^i\mu_i-\rho^\mu_{conv}=\rho^\mu-\rho^\mu_{conv}\geq0$, where $\rho^\mu$ is as Remark \ref{rmk:fatou}. This formulation could be linked to the gain that sharing may provide over some initial allocation. However, this is beyond our scope in this paper.
 \end{enumerate}
 	\end{Rmk}

 	We now state general results regarding $\tau_\rho$ in our framework. To this end, the following property of risk measures is required.
 	
 	\begin{Def}\label{def:iconv}
 		A risk measure $\rho\colon L^\infty\rightarrow\mathbb{R}$ is called $\mathcal{I}$-convex if $\rho\left(\sum_{i\in\mathcal{I}}X^i\mu_i \right)\leq\sum_{i\in\mathcal{I}}\rho\left(X^i\right)\mu_i$ for any  $\{X^i\}_{i\in\mathcal{I}}\in\cup_{X\in L^\infty}\mathbb{A}(X)$.
 	\end{Def}

 	\begin{Thm}\label{prp:arb}
 		We have the following for a risk measure $\rho$:
 		\begin{enumerate}
 			\item $\rho$ is $\mathcal{I}$-convex if and only if it is free of regulatory arbitrage.
 			\item If $\rho$ is a convex risk measure, then it is free of regulatory arbitrage. 
 			\item If $\rho$ is subadditive, then it is at most of finite regulatory arbitrage.
 			\item Let $\rho_1,\rho_2\colon L^\infty\rightarrow\mathbb{R}$ be risk measures such that $\rho_1\leq\rho_2$. If $\rho_1$ is of finite regulatory arbitrage, then $\rho_2$ is not of infinite regulatory arbitrage. Moreover, if  $\rho_2$ is of  infinite (or partially infinite) regulatory arbitrage, then so is $\rho_1$. 
 			\item If $\rho$ satisfies positive homogeneity, then $\tau_\rho(0)>0$ if and only if $\tau_\rho(0)=\infty$.
 			\item  If $\rho$ is loaded, then it is not of infinite regulatory arbitrage. If, in addition, it has the limitedness property, then it is of finite regulatory arbitrage. 
 		\end{enumerate}
 	\end{Thm}
 	
 	\begin{proof}
 		We note that as we consider finite risk measures, it holds that $\tau_\rho(X)=\infty$ if and only if $\rho^\mu_{conv}(X)=-\infty$. Then,
 		\begin{enumerate}
 			\item We assume that $\rho$ is $\mathcal{I}$-convex, and let $X\in L^\infty$. Then, $\rho(X)\leq\sum_{i\in\mathcal{I}}\rho(X^i)\mu_i$ for any $\{X^i\}_{i\in\mathcal{I}}\in\mathbb{A}(X)$. By taking the infimum over $\mathbb{A}(X)$, we obtain $\rho^\mu_{conv}(X)\leq\rho(X)\leq\rho^\mu_{conv}(X)$. For the converse, we obtain $\rho(X)=\rho^\mu_{conv}(X)\leq\sum_{i\in\mathcal{I}}\rho(X^i)\mu_i$ for any $\{X^i\}_{i\in\mathcal{I}}\in\mathbb{A}(X)$, which is $\mathcal{I}$-convexity. 
 			\item By Corollary \ref{crl:hull}, we have that $\rho^\mu_{conv}(X)=\rho(X),\:\forall\:X\in L^\infty$. As a direct consequence, we obtain that $\rho$ is free of regulatory arbitrage. 
 			\item  We begin with the claim that if $\rho$ is subadditive, then it is of partially infinite regulatory arbitrage if and only if it is of infinite regulatory arbitrage. By Proposition \ref{prp:propconv} and Remark \ref{rmk:prop}, we have that $\rho^\mu_{conv}$ is also subadditive and normalized. We only need to show that partially infinite regulatory arbitrage implies infinite regulatory arbitrage. Let $X\in L^\infty$ be such that $\tau_\rho(X)=\infty$. As $\rho$ is finite, it holds that $\rho^\mu_{conv}(X)=-\infty$. Let now $Y\in L^\infty$. We have that $\rho^\mu_{conv}(Y)\leq\rho^\mu_{conv}(X)+\rho^\mu_{conv}(Y-X)= -\infty$. Thus, $\rho$ is of infinite regulatory arbitrage. However, we have that $\tau_\rho(0)=\rho(0)-\rho^\mu_{conv}(0)=0<\infty$. Then, $\rho$ is not of infinite regulatory arbitrage, and, by the previous claim, it is not of partially infinite regulatory arbitrage either. Thus, $\rho$ is at most of finite regulatory arbitrage.
 			\item It is evident that, in this case, we have, by abuse of notation, $(\rho_1)^\mu_{conv}\leq(\rho_2)^\mu_{conv}$. If $\rho_1$ is of finite regulatory arbitrage, then $-\infty<(\rho_1)^\mu_{conv}(X)\leq(\rho_2)^\mu_{conv}(X),\:\forall\:X\in L^\infty$. Thus, $\rho_2$ is also of infinite regulatory arbitrage. If $\rho_2$ is of infinite regulatory arbitrage, then $(\rho_1)^\mu_{conv}(X)\leq(\rho_2)^\mu_{conv}(X)=-\infty,\:\forall\:X\in L^\infty$. Thus, $\rho_1$ is also of finite regulatory arbitrage. For partially infinite regulatory arbitrage, the reasoning is analogous.
 			\item We only need to prove the ``only if'' part because the converse is automatically obtained. As $\rho(0)=0$, $\tau_\rho(0)>0$ implies $\rho^\mu_{conv}(0)<0$. Then, there is $\{X^i\}_{i\in\mathcal{I}}\in\mathbb{A}(0)$ such that $\rho^\mu_{conv}(0)\leq\sum_{i\in\mathcal{I}}\rho(X^i)\mu_i<0$. As $\{\lambda X^i\}_{i\in\mathcal{I}}\in\mathbb{A}(0)\:\forall\:\lambda\in\mathbb{R}_+$, by the positive homogeneity of $\rho$, we obtain that \[\rho^\mu_{conv}(0)\leq\lim\limits_{\lambda\rightarrow\infty} \sum_{i\in\mathcal{I}}\rho(\lambda X^i)\mu_i=\lim\limits_{\lambda\rightarrow\infty}\lambda \sum_{i\in\mathcal{I}}\rho(X^i)\mu_i=-\infty.\] Hence, $\tau_\rho(0)=\rho(0)-\rho^\mu_{conv}(0)=\infty$.
 			\item By Proposition \ref{prp:propconv}, we have that $\rho^\mu_{conv}$ inherits loadedness and limitedness from $\rho$. The loadedness of $\rho$ implies normalization of $\rho^\mu_{conv}$, and therefore $\tau_\rho(0)=0$. Thus, $\rho$ is not of infinite regulatory arbitrage. If, in addition, $\rho$ is limited, then for any $X\in L^\infty$, we have that $\tau_\rho(X)\leq E[X]-\operatorname{ess}\inf X<\infty$. Hence, we obtain finite regulatory arbitrage for $\rho$.
 		\end{enumerate}
 	\end{proof}
 	
 	\begin{Rmk}
 	 A remarkable feature is that is possible to identify $\tau_\rho$ as a deviation measure in the sense of \cite{Rockafellar2006}, \cite{Rockafellar2013}, \cite{Righi2016}, \cite{Righi2018a}, and \cite{Righi2020}. For instance, the bound for $\tau_\rho$ in the proof of  (vi) is known as lower-range dominance for deviation measures.
  	\end{Rmk}
 	
 	\section{Examples}\label{sec:exm}
 	
In this section, we expose some concrete examples with specific choices for $\rho_\mathcal{I}$ in order to illustrate concepts from the paper. We start with VaR in the context of self-convolution. It was proved in \cite{Wang2016} that $VaR^\alpha$ is of infinite regulatory arbitrage. The following proposition adapts this claim to our framework. For this result, we are considering an atom-less probability space.
 
 \begin{Prp}\label{prp:var}
 	Let  $\alpha\in(0,1]$. Then $VaR^\alpha$ is of infinite regulatory arbitrage.
 \end{Prp}
 
 \begin{proof}
 	Let $\{i_j\}_{j=1}^{k+1}$ be members of $\mathcal{I}_\mu$ such that $\frac{1}{k}<\alpha$. Then $k>1$. Moreover, let $\{B_j,\:j=1,\cdots,k\}$ be a partition of $\Omega$ such that $\mathbb{P}(B_j)=\frac{1}{k}$ for any $j=1,\cdots,k$. We note that as $(\Omega,\mathcal{F},\mathbb{P})$ is atomless, such a partition always exists. For fixed $X\in L^\infty$ and some arbitrary real number $m>0$, let $\{X^i\}_{i\in\mathcal{I}}$ be defined as \[X^i(\omega)=\begin{cases*}
 	\dfrac{m(1-k1_{B_j}(\omega))}{(k-1)\mu_i},\:\text{for}\:i=i_j,j=1,\cdots,k,\\
 	\dfrac{X(\omega)}{\mu_{i_{k+1}}},\:\text{for}\:i=i_{k+1},\\
 	0,\:\text{otherwise}.\end{cases*}\]
 	for any $\omega\in\Omega$. Thus, $\{X^i\}_{i\in\mathcal{I}}\in\mathbb{A}(X)$ because for any $\omega\in\Omega$, the following is true:
 	\begin{align*}
 	\sum_{i\in\mathcal{I}}X^i(\omega)\mu_i&=\sum_{j=1}^k\left[ \dfrac{m(1-k1_{B_j}(\omega))}{(k-1)\mu_{i_j}}\right] \mu_{i_j}+\left[ \dfrac{X(\omega)}{\mu_{i_{k+1}}}\right] \mu_{i_{k+1}}\\
 	&=\dfrac{m}{k-1}\sum_{j=1}^k(1-k1_{B_j}(\omega))+X(\omega)\\
 	&=\dfrac{m}{k-1}(k-k)+X(\omega)=X(\omega).
 	\end{align*}
 	Furthermore, we note that for $i=i_j,\:j=1,\cdots,k$, we have that \[\mathbb{P}(X^i<0)=\mathbb{P}\left(1_{B_j}>\frac{1}{k}\right)=\mathbb{P}\left(1_{B_j}=1\right)=\mathbb{P}\left(B_j\right)=\frac{1}{k}<\alpha.\] Thus, $VaR^\alpha(X^i)<0$. In fact,  we have that $VaR^\alpha(-1_{B_j})=0$ and thus \[VaR^\alpha(X^i)=\frac{m}{(k-1)\mu_{i_j}}\left(kVaR^\alpha(-1_{B_j})-1 \right)=-\frac{m}{(k-1)\mu_{i_j}}<0.\]
 	As \[\sum_{i\in\mathcal{I}}VaR^\alpha(X^i)\mu_i=\sum_{j=1}^kVaR^\alpha(X^i) \mu_{i_j}+VaR^\alpha(X^{i_{k+1}}) \mu_{i_{k+1}}=VaR^\alpha(X)-\frac{mk}{k-1},\] $VaR^\alpha(X)<\infty$, and $m>0$ is arbitrary, we obtain that \[\rho^\mu_{conv}(X)\leq VaR^\alpha(X)-\lim\limits_{m\rightarrow\infty}\frac{mk}{k-1}=-\infty.\]
 	Hence, we conclude that $\tau_\rho(X)=\infty$ for any $X\in L^\infty$, which implies that $VaR^\alpha$ is of infinite regulatory arbitrage.
 \end{proof}
 
 Our following example is regarding to ES. This risk measure is governed by the parameter $\alpha$ that represents the tail level. It is clear that ES is non-increasing in $\alpha$. Consider a situation where each of our countable risk measures is the ES at distinct tail levels. For the finite case, this is studied in \cite{Embrechts2018}, for instance. In our framework, we have that the inf-convolution provides the same risk as the less conservative option. 
 
 \begin{Crl}
Let $\rho^i=ES^{\alpha^i},\:\alpha^i\in[0,1]\:\forall\:i\in\mathcal{I}$, and $\alpha=\sup_{i\in\mathcal{I}_\mu}\alpha^i$. Then $\rho^\mu_{conv}(X)=ES^\alpha(X),\:\forall\:X\in L^\infty$.
 \end{Crl}

\begin{proof}
The result arises as a direct consequence of (ii) in Theorem \ref{Thm:dualLI}.
\end{proof}

 		A relevant concept in the present context is the dilated risk measure, which is stable under inf-convolution and has a dilatation property with respect to the size of a position. In this particular situation, we can provide explicit solutions for optimal allocations even without law invariance. We now define this concept and extend some interesting related results to our framework.
 	
 	\begin{Def}\label{def:dila}
 		Let $\rho\colon L^\infty\rightarrow\mathbb{R}$ be a risk measure, and $\gamma>0$ be a real parameter. The dilated risk measure with respect to $\rho$ and $\gamma$ is a functional $\rho_\gamma\colon L^\infty\rightarrow\mathbb{R}$ defined as \begin{equation}
 		\rho_\gamma(X)=\gamma\rho\left(\frac{1}{\gamma}X\right).
 		\end{equation}
 	\end{Def}

 	\begin{Prp}\label{prp:dila}
 		We have that
 		\begin{enumerate}
 			\item $(\rho^\mu_{conv})_\gamma=\inf\left\lbrace \sum_{i\in\mathcal{I}}(\rho^i)_\gamma(X^i)\mu_i\colon\{X^i\}_{i\in\mathcal{I}}\in\mathbb{A}(X)\right\rbrace$ for any $\gamma>0$. 
 			\item Let $\rho$ be a convex risk measure and $\{\gamma^i>0\}_{i\in\mathcal{I}}$, such that $\sum_{i \in \mathcal{I}} \gamma^i \mu_i = \gamma $. If $\rho^i=\rho_{\gamma^i}\:\forall\:i\in\mathcal{I}_\mu$, then $\left\lbrace \frac{\gamma^i}{\gamma}X\right\rbrace_{\:i\in\mathcal{I}} $ is optimal for $X\in L^\infty$ and $\rho^\mu_{conv}=\rho_\gamma$. 
 		\end{enumerate}
 	\end{Prp}
 	
 	\begin{proof}
 		\begin{enumerate}
 			\item For any $\gamma>0$ and $X\in L^\infty$, we have that \begin{align*}
 			\inf\left\lbrace \sum_{i\in\mathcal{I}}(\rho^i)_\gamma(X^i)\mu_i\colon\{X^i\}_{i\in\mathcal{I}}\in\mathbb{A}(X)\right\rbrace&=\gamma\inf\left\lbrace \sum_{i\in\mathcal{I}}\rho^i\left(\frac{1}{\gamma}X^i\right)\mu_i\colon\{X^i\}_{i\in\mathcal{I}}\in\mathbb{A}(X)\right\rbrace\\
 			&=\gamma\inf\left\lbrace \sum_{i\in\mathcal{I}}\rho^i\left(Y^i\right)\mu_i\colon\{Y^i\}_{i\in\mathcal{I}}\in\mathbb{A}\left(\frac{1}{\gamma}X\right)\right\rbrace\\
 			&=\gamma\rho^\mu_{conv}\left(\frac{1}{\gamma}X\right)=(\rho^\mu_{conv})_\gamma(X).
 			\end{align*}
 			\item  We obtain from Theorem \ref{Thm:dualconv} that 
 			\[\alpha^{min}_{\rho^\mu_{conv}}(m)=  \sum_{i\in\mathcal{I}}\alpha^{min}_{\rho_{\gamma^i}}(m)\mu_i = \sum_{i\in\mathcal{I}}\gamma^i\alpha^{min}_{\rho}(m)\mu_i=\gamma\alpha^{min}_{\rho}(m)=\alpha^{min}_{\rho_\gamma}(m),\:\forall\:m\in ba.\] 
 			Thus, $\rho^\mu_{conv} = \rho_\gamma$. It is straightforward to note that $\left\lbrace \frac{\gamma^i}{\gamma}X\right\rbrace_{\:i\in\mathcal{I}}\in\mathbb{A}(X)$. We then obtain that
 			\[\rho^\mu_{conv}(X)\leq\sum_{i\in\mathcal{I}}\rho_{\gamma^i}\left(\frac{\gamma^i}{\gamma}X\right)\mu_i=\sum_{i\in\mathcal{I}}\gamma^i\rho\left(\frac{1}{\gamma^i}\frac{\gamma^i}{\gamma}X\right)\mu_i=\gamma\rho\left(\frac{1}{\gamma}X\right)=\rho_\gamma(X)\leq\rho^\mu_{conv}(X).\] Hence, $\left\lbrace \frac{\gamma^i}{\gamma}X\right\rbrace_{\:i\in\mathcal{I}} $ is optimal for $X\in L^\infty$.
 		\end{enumerate}
 	\end{proof}
 	
 	\begin{Rmk}
 It is evident that, for convex risk measures, $\alpha^{min}_{\rho_\gamma}=\gamma\alpha^{min}_\rho$. Moreover, a convex risk measure is coherent if and only if $\rho=\rho_\gamma$ pointwise for any $\gamma>0$. Under normalization, $\lim\limits_{\gamma\rightarrow\infty}\rho_\gamma$ defines the smallest coherent risk measure that dominates $\rho$. Further, we note that for any $X\in L^\infty$, $\left\lbrace \gamma^i\gamma^{-1}X\right\rbrace_{i\in\mathcal{I}}$ is $\mathcal{I}$-comonotone, which is in consonance with Theorem \ref{Thm:como} when the risk measures in $\rho_{\mathcal{I}}$ are law invariant. 
 	\end{Rmk}
 
	A typical example of dilated measure is the Entropic risk measure (Ent).  This is a Fatou-continuous, law-invariant, convex risk measure defined as $Ent^\gamma(X)=\frac{1}{\gamma}\log\left(E\left[e^{-\gamma X}\right]\right),\:\gamma\geq0$. Its acceptance set is defined as $\mathcal{A}_{Ent^\gamma}=\left\lbrace X\in L^\infty\colon E[e^{-\gamma X}]\leq 1 \right\rbrace$, and the penalty term is $\alpha^{\min}_{Ent^\gamma}(\mathbb{Q})=\frac{1}{\gamma}E\left[\frac{d\mathbb{Q}}{d\mathbb{P}}\log \left( \frac{d\mathbb{Q}}{d\mathbb{P}} \right) \right]$. It is clear that  $Ent_\gamma$ is a dilated risk measure with  $Ent_1$ as basis. In this case, we have the following result for inf-convolution.
	
	\begin{Crl}
Let $\rho^i=Ent_{\gamma^i},\:\gamma^i>0,\:\forall\:i\in\mathcal{I}_\mu$, and  $\gamma =\sum_{i \in \mathcal{I}} \gamma^i \mu_i$. Then $\rho^\mu_{conv}(X) = Ent_\gamma(X),\:\forall\:X\in L^\infty$ and  $\left\lbrace \frac{\gamma^i}{\gamma}X\right\rbrace_{\:i\in\mathcal{I}} $ is optimal for $X\in L^\infty$. 
	\end{Crl}

\begin{proof}
Direct from item (ii) in Proposition \ref{prp:dila}.
\end{proof}
 	
 	\bibliography{ref}
 	\bibliographystyle{elsarticle-harv}
 \end{document}